\documentclass[11pt,reqno]{amsart}
\usepackage{amsmath,amsfonts,amscd,amssymb,epsf,mathrsfs,color,enumerate}
\newtheorem{theorem}{Theorem}[section]

\newtheorem{lemma}[theorem]{Lemma}
\newtheorem{corollary}[theorem]{Corollary}
\theoremstyle{definition}

\theoremstyle{remark} \newtheorem{remark}[theorem]{Remark}

\numberwithin{equation}{section}
\newcommand{\field}[1]{\ensuremath{\mathbb{#1}}}

\newcommand{\RR}{\field{R}}

\newcommand{\delb}{\bar\partial}

\newcommand{\up}{{\mathbb{U}}}

\newcommand{\C}{{\mathbb{C}}}

\newcommand{\bk}{\backslash}
\newcommand{\Ga}{\Gamma}

\newcommand{\pa}{\partial}

\newcommand{\ov}{\overline}

\newcommand{\vep}{\varepsilon}

\newcommand{\di}{\displaystyle}

\begin{document}

\title[Liouville action and Holography]{Liouville action and Holography on quasi-Fuchsian deformation spaces}
\author{Jinsung Park}
\address{School of Mathematics, Korea Institute for Advanced Study, 207-43, Hoegiro 85, Dong-daemun-gu, Seoul, 130-722, Korea}
\email{jinsung@kias.re.kr}
\author{Lee-Peng Teo}
\address{Department of Mathematics, Xiamen University Malaysia, Jalan Sunsuria, Bandar Sunsuria, 43900, Selangor, Malaysia}
\email{lpteo@xmu.edu.my}
\maketitle

\begin{abstract}
We study the Liouville action for quasi-Fuchsian groups with parabolic and elliptic elements. In particular,
when the group is Fuchsian, the contribution of elliptic elements to the classical Liouville action is derived in terms of the Bloch-Wigner functions. We prove the first and second variation formulas for the classical Liouville action  on the quasi-Fuchsian deformation space.
We prove an equality expressing the holography principle, which relates the Liouville action and the renormalized volume for quasi-Fuchsian groups with parabolic and elliptic elements.
We also construct the potential functions of the K\"ahler forms corresponding to the Takhtajan-Zograf metrics associated to the elliptic elements in the quasi-Fuchsian groups.
\end{abstract}


\section{Introduction}

In this paper, we study the Liouville action for quasi-Fuchsian groups with parabolic and elliptic elements. We also prove an equality expressing the \emph{holography principle}, that is, a relationship between the Liouville action and the renormalized volume for quasi-Fuchsian groups. This work can be considered as a continuation of the previous papers \cite{2}, \cite{1} where we restricted types of the quasi-Fuchsian groups. In \cite{1}, we considered the case of quasi-Fuchsian
groups with only parabolic elements. A main novelty of this paper is the derivation of the contributions of the elliptic elements to the Liouville action and the holography principle.
Interestingly, such new contributions are given in terms of the Bloch-Wigner functions, which are variants of the dilogarithm functions.

Now we introduce some notations to explain main results of this paper. Let $\Gamma$ be a quasi-Fuchsian group in $\text{PSL}(2,\mathbb{C})$ with
its region of discontinuity $\Omega=\Omega_1\sqcup \Omega_2$.
Let $X\simeq\Gamma\backslash\Omega_1$ and $Y\simeq \Gamma\backslash\Omega_2$ be the
corresponding Riemann surfaces with opposite orientations. We allow $\Gamma$ to have some parabolic elements and elliptic elements, so that the Riemann surfaces
$X$ and $Y$ have possibly punctures and ramification points.
We also assume some topological conditions so that $X$ and $Y$ can be equipped with hyperbolic metrics.
The first part of this paper deals with the Liouville action for these Riemann surfaces.
Since these Riemann surfaces have punctures and ramification points, we need to check behaviours of the metrics  near these points
in order to define the Liouville action. By some estimates, we show that the Liouville action can be defined for the hyperbolic metric for Riemann surfaces with punctures and ramification
points. The Liouville action defined for the hyperbolic metric is called \emph{classical}.  The new results of the first part are descriptions of the classical Liouville
action over the deformation space $\mathfrak{D}(\Gamma)$ of the quasi-Fuchsian groups.
For the precise definitions
of these, see the subsection \ref{ss:coho} and the paragraph near equality \eqref{eq:deform-sp} respectively.
In particular, we derive the contribution
of the elliptic elements to the classical Liouville action when the group $\Gamma$ is Fuchsian and obtain
the first and second variational formulas
of the classical Liouville action over $\mathfrak{D}(\Gamma)$. The following theorem is given at Theorems \ref{classical_Liouville},
\ref{firstvariation}, and \ref{thm:second-var} with more explanations.

\begin{theorem}\label{thm-int:var}
When $\Gamma$ is Fuchsian, the classical Liouville action $S$ is given by
\begin{align*}
S=8\pi \chi(X)+4\sum_{j=1}^r  D\left( e^{\frac{2\pi i}{m_j}} \right).
\end{align*} Here $\chi(X)$ is given by \eqref{e:def-chi}, $D(z)$ denotes the Bloch-Wigner function given in \eqref{e:def-BW}, and  $m_j$ denotes the ramification index for $j=1,\ldots, r$.
For the classical Liouville action $S$ on the deformation space $\mathfrak{D}(\Gamma)$,
\begin{align*}
\pa {S}=\vartheta,\qquad \bar{\pa}\pa S=-2i\omega_{WP}.
\end{align*}
Here $\vartheta$ denotes the $(1,0)$-form over $\mathfrak{D}(\Gamma)$ corresponding to
the holomorphic quadratic differential $2\varphi_{zz}-\varphi_z^2$ for the hyperbolic metric $e^{\varphi(z)}|dz|^2$
and $\omega_{WP}$ denotes the Weil-Petersson symplectic form over $\mathfrak{D}(\Gamma)$.
\end{theorem}

Although the above results of the Liouville action hold only along the hyperbolic metrics,
the Liouville action can be defined for
any metrics {which have same  singular behaviours as}  the hyperbolic metric near punctures and ramifications points. We denote the set of such metrics over $X\sqcup Y$
by $\mathcal{CM}(X\sqcup Y)$. By the holography principle,
the Liouville action of a metric is expected to have a relationship with the renormalized volume of a hyperbolic 3-manifold which has the given pair of Riemann surfaces $X\sqcup Y$ as conformal boundaries \cite{Kras-Sch08}, \cite{2}, \cite{1}. Such a hyperbolic 3-manifold $M$ is given by the quotient of the hyperbolic 3-space by the quasi-Fuchsian group $\Gamma$.
Since we allow  punctures and ramification points for $X$ and $Y$, correspondingly the manifold $M$ has rank one cusps and conical singularities
of codimension 2. The second part of this paper deals with the proof of this relationship expected from the holography principle. The main task for this is to elaborate the analysis for contribution from rank one cusps and conical singularities. The case of rank one cusps was also discussed briefly in \cite{1} and these singular structures did not produce any additional terms, but the conical singularities of codimension 2 produce some additional terms expressed by the Bloch-Wigner functions. A metric $e^{\phi(z)}|dz|^2\in \mathcal{CM}(X\sqcup Y)$ is used to
renormalize the hyperbolic volume of $M$ near conformal boundaries $X\sqcup Y$ as in other works  \cite{Kras00}, \cite{Kras-Sch08}, \cite{2}, but we also need other regularization process
near rank one cusps and conical singularities of codimension 2. In this way, we can regularize the hyperbolic volume of $M$ and define the Einstein-Hilbert
action $\mathcal{E}[\phi]$ by \eqref{eq:def-EH}, which is the same as $-4$ times of the renormalized volume.  The following theorem is given at Theorem \ref{thm:holography} with more explanations.

\begin{theorem}\label{thm-int:hol}
For $e^{\phi(z)}|dz|^2 \in \mathcal{CM}(X\sqcup Y)$,
\begin{align*}
\mathcal{E}[\phi]=S[\phi]-\iint\limits_{X\sqcup Y} e^{\phi(z)}d^2z-8\pi\chi(X)\log 2 -4\sum_{j=1}^r D\left( e^{\frac{2\pi i}{m_j} }\right).
\end{align*}
\end{theorem}
In the above equality,
the terms given by the Bloch-Wigner functions are derived from the conical singularities of codimension 2 of $M$. Correspondingly, as we stated in Theorem \ref{thm-int:var}, the exactly same terms appear as the contribution of the elliptic elements to the classical Liouville action $S$ when $\Gamma$ is Fuchsian.
A hyperbolic manifold with conical singularities of codimension 2 is called  a hyperbolic manifold with particles in  \cite{Kras-Sch07}.
Hence, it seems to be interesting to understand the terms given by the Bloch-Wigner functions in Theorem \ref{thm-int:hol} from the viewpoint of \cite{Kras-Sch07}.

Recently Takhtajan and Zograf introduced a metric associated to a ramification point over a Riemann surface in \cite{TZ17}. The precise definition for this is given in Section \ref{s:elliptic}. For a quasi-Fuchsian group with elliptic elements,
this metric can be defined for each pair of ramification points in $X\sqcup Y$ determined by an elliptic element.
Considering Theorems \ref{thm-int:var} and \ref{thm-int:hol},  one would naturally ask about the construction of potential functions
of the K\"ahler forms  corresponding to these metrics on the quasi-Fuchsian deformation space $\mathfrak{D}(\Gamma)$.
Employing machinery to prove aforementioned theorems, we construct such potential functions in two ways. The first one is  constructed analytically from the hyperbolic metric over $X\sqcup Y$, and the second one is  constructed geometrically from $M$. The precise definitions
of these and the corresponding results are given in Section \ref{s:elliptic}.

Now we explain the structure of this paper. In Section \ref{sec2}, we construct the Liouville action over quasi-Fuchsian deformation space $\mathfrak{D}(\Gamma)$ and derive the elliptic contribution to the classical Liouville action
when $\Gamma$ is Fuchsian. In Section \ref{s:variation}, we prove the results for variation formulas of the classical Liouville action.
In Section \ref{s:holography}, we prove the equality relating the Liouville action to the Einstein-Hilbert action. In Section \ref{s:elliptic},
we construct potential functions of the Takhtajan-Zograf metric associated to elliptic elements in $\Gamma$. In Appendix \ref{a1}, we provide some
estimates for the hyperbolic metric near punctures and ramification points.

\subsection*{Acknowledgements}
The work of J. Park was partially supported by Samsung Science and Technology Foundation under Project Number SSTF-BA1701-02. We would like to thank L. Takhtajan and K. Krasnov for giving constructive comments to the first draft of this paper.

\section{The Liouville action functional on quasi-Fuchsian deformation spaces}\label{sec2}

Consider a Riemann surface of finite type $X$, with genus $g$, $n\geq 0$ punctures, and $r\geq 0$ ramified points with ramification indices $m_1$, $m_2$, $\ldots$, $m_r$, where $2\leq m_1\leq m_2\leq \ldots\leq m_r$. We say that the Riemann surface is of type $(g, n; m_1, m_2,\ldots, m_r)$. The characteristic of the surface $X$ is given by
\begin{align}\label{e:def-chi}
\chi(X)=2g-2+n+\sum_{i=1}^r\left(1-\frac{1}{m_i}\right).
\end{align}In this work, we assume that $\chi(X)>0$ so that the surface $X$ is a hyperbolic surface. Then we can realize $X$ as a quotient space $\Gamma\bk\mathbb{U}$, where $\mathbb{U}$ is the upper half plane, and $\Gamma$ is a Fuchsian group of the first kind. Here $\Gamma$ is a finitely generated cofinite discrete subgroup of $\text{PSL} (2, \mathbb{R})$ which has a standard representation with $2g$ hyperbolic generators $\alpha_1, \beta_1, \ldots, \alpha_g, \beta_g$, $n$ parabolic generators $\kappa_1, \ldots, \kappa_n$, and $r$ elliptic generators $\tau_1, \ldots, \tau_r$ of orders $m_1, \ldots, m_r$ respectively, satisfying the relation
\begin{align*}
\alpha_1\beta_1\alpha_1^{-1}\beta_1^{-1}\ldots \alpha_g\beta_g\alpha_g^{-1}\beta_g^{-1}\kappa_1\ldots\kappa_n\tau_1\ldots\tau_r=I,
\end{align*}where $I$ is the identity element in $\Gamma$.
 The group is normalized by prescribing three of the fixed points of the generators. For example, if $g\geq 1$, the group is normalized so that the attracting and repelling fixed points of $\alpha_1$ are 0 and $\infty$ respectively, and the attracting fixed point of $\beta_1$ is 1.

 Assume that $3g-3+n+r>0$, so that the moduli space of $X$ has positive dimension.
 In this section, we discuss how to construct the Liouville action on the quasi-Fuchsian deformation spaces. The construction is similar to our previous works \cite{2}, \cite{1} for quasi-Fuchsian deformation spaces of compact Riemann surfaces and   Riemann surfaces with punctures, but we have to take care of the elliptic elements.

Given a marked, normalized, quasi-Fuchsian group $\Gamma$, its region of discontinuity $\Omega$ has two invariant components $\Omega_1$ and $\Omega_2$ separated by a quasi-circle $\mathcal{C}$. Let $X\simeq\Gamma\backslash\Omega_1$ and $Y\simeq \Gamma\backslash\Omega_2$ be the
corresponding marked  Riemann surfaces with opposite orientations.  We say that the quasi-Fuchsian group is of type $(g, n; m_1, m_2,\ldots, m_r)$ if both the Riemann surfaces $X$ and $Y$ are of type $(g, n; m_1, m_2,\ldots, m_r)$. There exists a quasiconformal homeomorphism $J_1$ of $\hat{\mathbb{C}}$ such that

\begin{enumerate}
\item[\textbf{QF1}]$J_1$ is holomorphic on $\mathbb{U}$ and $J_1(\mathbb{U})=\Omega_1$, $J_1(\mathbb{L})=\Omega_2$, $J_1(\mathbb{R})=\mathcal{C}$, where $\mathbb{U}$ and $\mathbb{L}$ are respectively the upper and lower half planes.
\item[\textbf{QF2}] $J_1$ fixes $0, 1$ and $\infty$.
\item[\textbf{QF3}] $ \Gamma_1=J_1^{-1}\circ \Gamma\circ J_1$ is a marked, normalized Fuchsian group.
\end{enumerate}

This implies that $X\simeq \Gamma_1\backslash \mathbb{U}$. There is also a quasiconformal homeomorphism $J_2$ of $\hat{\mathbb{C}}$, holomorphic on $\mathbb{L}$, with a Fuchsian group $\Gamma_2=J_2^{-1}\circ \Gamma\circ J_2$ so that $Y\simeq \Gamma_2\backslash\mathbb{L}$. The hyperbolic metric $\displaystyle e^{\phi_{\text{hyp}}(z)}|dz|^2$ on $\Omega=\Omega_1\sqcup \Omega_2$ is given explicitly  by
\begin{align}\label{eq0630_1}
\rho(z)=e^{\phi_{\text{hyp}}(z)}=\frac{\left|\left(J_i^{-1}\right)_z(z)\right|^2}{\left|\text{Im}\,\left(J_i^{-1}(z)\right)\right|^2},\quad\text{if}\;\;z\in\Omega_i,\quad i=1, 2.
\end{align}
This is a pull-back by the map $J^{-1}:\Omega_1 \sqcup \Omega_2\rightarrow \mathbb{U}\sqcup\mathbb{L}$ of the hyperbolic metric on  $\mathbb{U}\sqcup\mathbb{L}$, where $\displaystyle J|\mathbb{U} =J_1|\mathbb{U}$ and $\displaystyle J|\mathbb{L} =J_2|\mathbb{L}$.

Denote by $\mathfrak{D}(\Gamma)$ the deformation space of the quasi-Fuchsian group $\Gamma$. It is a complex manifold of dimension $6g-6+2n+2r$. It is known that
\begin{align}\label{eq:deform-sp}
\mathfrak{D}(\Gamma) \simeq \mathfrak{T}(\Gamma_1)\times {\mathfrak{T}}(\Gamma_2)
\end{align}
where $\mathfrak{T}(\Gamma_i)$ is the Teichm\"uller space of $\Gamma_i$ for $i=1,2$.
For details about the definition of $\mathfrak{D}(\Gamma)$, we refer the readers to \cite{2}. The  holomorphic tangent space of $\mathfrak{D}(\Gamma)$ at the origin is identified with  $\Omega^{-1,1}(\Gamma)$ -- the complex vector space of harmonic Beltrami differentials. The Weil-Petersson K\"ahler form on $\mathfrak{D}(\Gamma)$ is induced by the pairing
\begin{equation}\label{e:def-WP pairing}
\langle \mu, \nu\rangle_{\text{WP}} =\iint\limits_{X\sqcup Y}\mu(z) \overline{\nu(z)}\rho(z)d^2z
\end{equation}
for $\mu, \nu\in \Omega^{-1,1}(\Gamma)$.

In the following two subsections, we present the construction of the Liouville action over
$\mathfrak{D}(\Gamma)$. For this, basically we follow the construction in \cite{2} where
the quasi-Fuchsian group $\Gamma$ is assumed to have no parabolic and elliptic elements. Since these type elements are allowed in this paper,
we take care of these in the construction and explicate the difference when they contribute nontrivially. To avoid much repetition, we will skip some part of the construction and refer to the section 2 of \cite{2} for details.

\subsection{Homology construction}

Start with a marked Fuchsian group $\Gamma$ of type $(g, n; m_1, m_2,\ldots, m_r)$, the double homology complex $\mathsf{K}_{\bullet, \bullet}$ is defined as $\mathsf{S}_{\bullet}\otimes_{\mathbb{Z}\Gamma}\mathsf{B}_{\bullet}$, a tensor product over the integral group ring $\mathbb{Z}\Gamma$, where $\mathsf{S}_{\bullet}=\mathsf{S}_{\bullet}(\mathbb{U})$ is the singular chain complex of $\mathbb{U}$ with the differential $\partial'$, considered as a right $\mathbb{Z}\Gamma$-module, and $\mathsf{B}_{\bullet}=\mathsf{B}_{\bullet}(\mathbb{Z}\Gamma)$ is the standard bar resolution complex for $\Gamma$ with differential $\partial''$. The associated total complex $\text{Tot} \;\mathsf{K}$ is equipped with the total differential $\partial=\partial'+(-1)^p \partial''$ on $\mathsf{K}_{p,q}$.

There is a standard choice of the fundamental domain
$F\subseteq \mathbb{U}$ for $\Gamma$ as a  non-Euclidean polygon with
$4g+2n+2r$ edges labeled by $a_k, a_k', b_k', b_k;1\leq  k\leq g$,
$c_i, c_i'; 1\leq i\leq n$ and $d_j, d_j'; 1\leq j\leq r$ satisfying
$\displaystyle  \alpha_k\left(a_k'\right)=a_k,\,\beta_k\left(b_k'\right)=b_k, \,1\leq k\leq g$,
$\displaystyle\kappa_i\left(c_i'\right)=c_i,\, 1\leq i\leq n$ and $\displaystyle\tau_j\left(d_j'\right)=d_j,\, 1\leq j\leq r$. The orientation of the edges
is chosen such that
\begin{equation}\label{eq1}
\pa' F=\sum_{k=1}^{g}(a_k+b_k'-a_k'-b_k)+\sum_{i=1}^{n}(c_i-c_i')+\sum_{j=1}^{r}(d_j-d_j').
\end{equation}
Set $\pa' a_k=a_k(1)-a_k(0),\, \pa' b_k=b_k(1)-b_k(0),\,
\pa'c_i=c_i(1)-c_i(0), \,\pa'd_j=d_j(1)-d_j(0)$, so that $a_k(0)=b_{k-1}(0); 2\leq k \leq g
$, $c_i(0)=c_{i-1}'(0),\,
2\leq i\leq n $, $c_1(0)= b_g(0)$,  $a_1(0)=d_r'(0) $. The relations
between the vertices of $F$ and the generators of $\Gamma$ are the
following: $\displaystyle \alpha_k^{-1}\left(a_k(0)\right)=b_k(1),\, \beta_k^{-1}\left(b_k(0)\right)=
a_k(1),\,1\leq k\leq g$; $ \displaystyle \gamma_k\left(b_k(0)\right)=b_{k-1}(0),\;2\leq
k\leq g$, $\displaystyle\gamma_1\left(b_1(0)\right)=a_1(0)$; $\displaystyle\kappa_i^{-1}\left(c_i(0)\right)=c_{i+1}(0)$, $1\leq i\leq n-1$;
$\displaystyle\tau_j^{-1}\left(d_j(0)\right)=d_{j+1}(0)$, $1\leq j\leq r-1$. Here $\gamma_k=[\alpha_k, \beta_k]=\alpha_k\beta_k\alpha_k^{-1}\beta_k^{-1}$.

\begin{figure}[h]
\epsfxsize=0.6\linewidth \epsffile{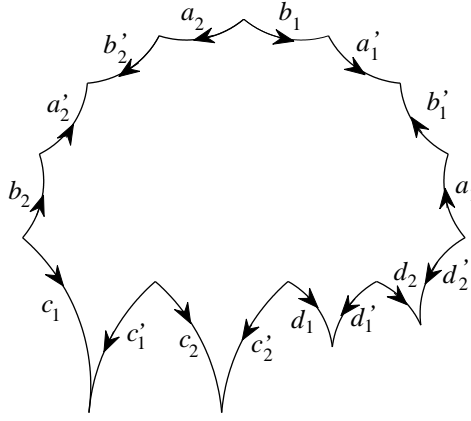} \caption{\label{f1} The convention of the fundamental domain $F$.}\end{figure}

 According to the isomorphism
$\mathsf{S}_{\bullet}\simeq\mathsf{K}_{\bullet,0}$, the fundamental domain $F$ is
identified with $F \otimes [\;] \in \mathsf{K}_{2,0}$. We have $\pa'' F
= 0$, and it follows from \eqref{eq1} that
\begin{equation*}\begin{split}
\pa' F = &\sum_{k=1}^g \left( \beta_k^{-1}\left(b_k\right) - b_k -
\alpha_k^{-1}\left(a_k\right) +a_k\right)-\sum_{i=1}^n\left(\kappa_i^{-1}\left(c_i\right)-c_i\right)
\\&-\sum_{j=1}^r\left(\tau_j^{-1}\left(d_j\right)-d_j\right) \\=& \pa'' L,\end{split}
\end{equation*}
where $L\in \mathsf{K}_{1,1}$ is given by
\begin{equation} \label{L}
L= \sum_{k=1}^g \left(b_k\otimes\left[ \beta_k\right] -a_k\otimes\left[ \alpha_k\right]\right)-\sum_{i=1}^n
c_i\otimes\left[\kappa_i\right]-\sum_{j=1}^r
d_j\otimes\left[\tau_j\right].
\end{equation}
There exists $V\in\mathsf{K}_{0,2}$ such that
\begin{align}\label{eq0914_1}
 \pa ' L=\pa'' V-\sum_{i=1}^n x_i \otimes [\kappa_i]-\sum_{j=1}^r z_j \otimes [\tau_j].
\end{align}
Here $x_i=c_i(1)$ are representatives of the punctures of the
Riemann surface $X=\Gamma\bk\mathbb{U}$ on $\mathbb{R}\cup \{\infty\}$
and $z_j=d_j(1)$ are representatives of the ramification points of $X$ on $\mathbb{U}$. In the presence of ramification points, the expression for $V$ is much more complicated. One can verify that it is given by
\begin{equation} \label{V}\begin{split}
V &= \sum_{k=1}^g \left(a_k(0)\otimes[ \alpha_k| \beta_k] -
b_k(0)\otimes\left[\beta_k|\alpha_k\right] +
b_k(0)\otimes\left[\gamma_k^{-1}|\alpha_k\beta_k\right]\right)\\ & \;\;\;
-\sum_{k=1}^{g-1}
b_g(0)\otimes\left[\gamma_g^{-1}\ldots\gamma_{k+1}^{-1}|\gamma_k^{-1}\right]\\&+\sum_{i=1}^{n-1}c_1(0)\otimes\left[\kappa_1
\cdots\kappa_i|\kappa_{i+1}\right] +\sum_{j=0}^{r-1}c_1(0)\otimes\left[\kappa_1\ldots\kappa_n\tau_1
\cdots\tau_j|\tau_{j+1}\right]. \end{split}
\end{equation}
Define $$\Sigma = F + L - V.$$ Then
\[
\pa \Sigma = -\sum_{i=1}^n x_i\otimes[\kappa_i]-\sum_{j=1}^r z_j\otimes [\tau_j].
\]
When $\Gamma$ contains parabolic or elliptic elements, $\Sigma$ is not a cycle.

Finally, we also define $W$ in the following way. Let $P_k$ be a $\Gamma$-contracting path (see \cite{2} for the definition of $\Gamma$-contractible) connecting 0 to $b_k(0)$. Then
\begin{equation} \label{W}\begin{split}
W =& \sum_{k=1}^g \left(P_{k-1}\otimes\left[\alpha_k|\beta_k\right] -
P_k\otimes\left[\beta_k|\alpha_k\right] +
P_k\otimes\left[\gamma_k^{-1}|\alpha_k\beta_k\right]\right)\\ &
-\sum_{k=1}^{g-1}
P_g\otimes\left[\gamma_g^{-1}\ldots\gamma_{k+1}^{-1}| \gamma_k^{-1}\right]+\sum_{i=1}^{n-1}P_g\otimes\left[\kappa_1\ldots\kappa_i|\kappa_{i+1}\right]\\
&+\sum_{j=0}^{r-1}P_g\otimes\left[\kappa_1\ldots\kappa_n\tau_1\ldots\tau_j|\tau_{j+1}\right]. \nonumber
\end{split}\end{equation}

If $\Gamma$ is a quasi-Fuchsian group, let $\Gamma_1$ be the Fuchsian group such that $\Gamma_1=J_1^{-1}\circ \Gamma\circ J_1$. The double complex associated with $\Omega=\Omega_1\sqcup \Omega_2$ and the group $\Gamma$  is a push-forward by the map $J_1$ of the double complex associated with $\mathbb{U}\sqcup\mathbb{L}$ and the group $\Gamma_1$.
\begin{align*}
\Sigma_1 =&J_1(\Sigma) =F_1 +L_1 -V_1,\\
 \Sigma_2=&J_1\left(\bar{\Sigma}\right)=F_2+L_2-V_2,
\end{align*}where $F_1=J_1(F)$, $F_2=J_1(\bar{F})$, $L_1=J_1(L)$, $L_2=J_1(\bar{L})$, $V_1=J_1(V)$, $V_2=J_1(\bar{V})$. Note that $\Sigma_1-\Sigma_2$ is a cycle only when there is no elliptic elements in $\Gamma$. In the general case we consider in this paper, $\Sigma_1-\Sigma_2$ is not a cycle.

\subsection{Cohomology construction}\label{ss:coho}

The corresponding double complex in cohomology $\mathsf{C}^{\bullet, \bullet}$ is defined as $\mathsf{C}^{p, q}=\text{Hom}_{\mathbb{C}}\left(\mathsf{B}_q, \mathsf{A}^p\right)$, where $\mathsf{A}^{\bullet}$ is the complexified de Rham complex on $\Omega=\Omega_1\sqcup\Omega_2$. The associated total complex $\text{Tot} \mathsf{C}$ is equipped with the total differential $D=d+(-1)^p\delta$ on $\mathsf{C}_{p, q}$, where $d$ is the de Rham differential and $\delta$ is the group coboundary. The natural pairing $\langle\; , \; \rangle$ between $\mathsf{C}^{p, q}$ and $\mathsf{K}_{p, q}$ is given by the integration over chains.

Let $\varphi=\phi_{\text{hyp}}$, where $e^{\phi_{\text{hyp}(z)}}|dz|^2$ is the hyperbolic metric on $\Omega$ given by \eqref{eq0630_1}. Denote by $\mathcal{CM}(X\sqcup Y)$  the space of conformal metrics on $\Gamma\backslash\Omega=X\sqcup Y$ satisfying certain regularity conditions at the parabolic and elliptic fixed points of $\Gamma$. That is, every $ds^2\in \mathcal{CM}(X\sqcup Y)$ is represented as $ds^2=e^{\phi(z)}|dz|^2$,  where $\phi$ is a smooth function on $\Omega$ satisfying
$$\phi\circ\gamma+\log|\gamma'|^2=\phi\quad\quad\forall \;\gamma\in\Gamma,$$ and
 $$\phi(z)-\varphi(z)=O(1)$$ as $z$ approaches the parabolic and elliptic fixed points of $\Gamma$. Since $J^{-1}$ is univalent on $\Omega_1\sqcup\Omega_2$, from \eqref{eq0630_1}, we find that $\varphi(z)$ is regular when $z$ approaches the elliptic fixed points, and so does $\phi(z)$.

 The Liouville action  is a function on the space of conformal metrics. Its construction is as follows.
Starting with the 2-form
\begin{equation} \label{omega}
\omega[\phi] = \left(|\phi_{z}|^2 +e^{\phi}\right)dz\wedge d\bar{z}\in \mathsf{C}^{2,0},
\end{equation}
we have
\begin{equation*}
\delta\omega[\phi]=d\check{\theta}[\phi],
\end{equation*}
where $\check{\theta}[\phi] \in \mathsf{C}^{1,1}$ is given explicitly by
\begin{equation} \label{theta}
\check{\theta}_{\gamma^{-1}}[\phi] = \left(\phi -
\frac{1}{2}\log|\gamma'|^2-2\log 2-\log|c(\gamma)|^2\right) \left(\frac{\gamma''}{\gamma'} dz -
\frac{\ov{\gamma''}}{\ov{\gamma'}}d\bar{z}\right).
\end{equation}Here $c(\gamma)$ is the element $c$ in the  linear fractional transformation $\di\gamma=\begin{pmatrix} a & b\\c & d\end{pmatrix}.$ Notice that $\check{\theta}_{\gamma^{-1}}[\phi]=0$ if $c(\gamma)=0$.

Next, set
\begin{equation*}
\check{u}=\delta \check{\theta}[\phi] \in \mathsf{C}^{1,2}.
\end{equation*}
From the definition of $\check{\theta}$ and $\delta^2=0$, it follows that the
1-form $\check{u}$ is closed. An explicit calculation gives
\begin{equation} \label{u}\begin{split}
\check{u}_{\gamma_1^{-1},\gamma_2^{-1}}= & -\left(\frac{1}{2}\log|\gamma_1'|^2+\log\frac{|c(\gamma_2)|^2}{|c(\gamma_2\gamma_1)|^2}\right)
\left(\frac{\gamma_2''}{\gamma_2'}\circ\gamma_1\, \gamma_1'\, dz -
\frac{\ov{\gamma_2''}}{\ov{\gamma_2'}}\circ\gamma_1\, \ov{\gamma_1'}\,d\bar{z}\right) \\
 & + \left(\frac{1}{2}\log|\gamma_2'\circ\gamma_1|^2+\log\frac{|c(\gamma_2\gamma_1)|^2}{|c(\gamma_1)|^2}\right) \left(\frac{\gamma_1''}{\gamma_1'} dz -
\frac{\ov{\gamma_1''}}{\ov{\gamma_1'}}d\bar{z}\right). \nonumber
\end{split}\end{equation}

\begin{remark}
Notice that for a linear transformation $\gamma$,
\begin{align}\label{eq11_14_1}
-2c(\gamma)=\frac{\gamma''(z)}{\left(\gamma'(z)\right)^{\frac{3}{2}}}.
\end{align}
Hence, $\check{\theta}_{\gamma^{-1}}[\phi]$ and $\check{u}_{\gamma_1^{-1},\gamma_2^{-1}}$ can be rewritten as
\begin{equation} \label{theta2}
\check{\theta}_{\gamma^{-1}}[\phi] = \left(\phi
 -\log\left|\frac{\gamma''}{\gamma'}\right|^2\right) \left(\frac{\gamma''}{\gamma'} dz -
\frac{\ov{\gamma''}}{\ov{\gamma'}}d\bar{z}\right),
\end{equation}
\begin{equation} \label{u2}\begin{split}
\check{u}_{\gamma_1^{-1},\gamma_2^{-1}}= & -\left( \log\left|\frac{\gamma_2''}{\gamma_2'}\circ\gamma_1\right|^2 \right)
\left(\frac{\gamma_2''}{\gamma_2'}\circ\gamma_1\, \gamma_1'\, dz -
\frac{\ov{\gamma_2''}}{\ov{\gamma_2'}}\circ\gamma_1\, \ov{\gamma_1'}\,d\bar{z}\right) \\
&+ \left(\log\left|\frac{(\gamma_2\circ \gamma_1)''}{(\gamma_2\circ\gamma_1)'}\right|^2\right)
\left(\frac{(\gamma_2\circ\gamma_1)''}{(\gamma_2\circ\gamma_1)'}  dz -
\frac{\ov{(\gamma_2\circ\gamma_1)''}}{\ov{(\gamma_2\circ\gamma_1)'}} d\bar{z}\right)\\
 &-\left(\log\left|\frac{\gamma_1''}{\gamma_1'}\right|^2\right) \left(\frac{\gamma_1''}{\gamma_1'} dz -
\frac{\ov{\gamma_1''}}{\ov{\gamma_1'}}d\bar{z}\right).
\end{split}\end{equation}
\end{remark}

For $e^{\phi(z)}|dz|^2\in \mathcal{CM}(X\sqcup Y)$, the Liouville action is defined as
\begin{align}\label{Liouvilleaction}
S[\phi]=&\frac{i}{2}\Bigl(\langle \omega[\phi], F_1-F_2\rangle-\langle\check{\theta}[\phi], L_1-L_2\rangle+\langle \check{u}, W_1-W_2\rangle\Bigr),
\end{align}
where $W_1=J_1(W)$ and $W_2=J_1(\bar{W})$. When $e^{\phi}=e^{\varphi}$ is the hyperbolic metric,  $S[\varphi]$ is well-defined by Theorem \ref{thm2}. For a conformal metric $e^{\phi(z)}|dz|^2\in  \mathcal{CM}(X\sqcup Y)$, $S[\phi]$ is also well-defined since
$\phi(z)-\varphi(z)=O(1)$ as $z$ approaches the parabolic and elliptic fixed points.

\begin{remark}
When $\Gamma$ contains elliptic elements,  $\Sigma_1-\Sigma_2$ is not a cycle. Hence, it is not clear that the Liouville action defined above is independent of the choice of fundamental domains.
By Theorem \ref{thm:holography},  the Liouville action is equal to the renormalized volume of the corresponding hyperbolic three-manifold defined by $\Gamma$ up to some constants. This shows that the Liouville action is indeed independent of the choice of fundamental domain.
\end{remark}

\begin{remark}
In \cite{2}, the Liouville action was originally defined by
\begin{align}\label{e:Li1}
S[\phi]=\frac{i}{2}\Bigl(\langle \omega[\phi], F_1-F_2\rangle-\langle\check{\theta}[\phi], L_1-L_2\rangle+\langle \check{\Theta}, V_1-V_2\rangle\Bigr)
\end{align}
where $\check{\Theta}\in \mathsf{C}^{0,2}$ with $d\check{\Theta}=\check{u}$.
It can be shown that the definition in \eqref{e:Li1} is equivalent to the definition in \eqref{Liouvilleaction} by
\begin{align*}
\langle \check{u}, W_1-W_2\rangle=&\langle d\check{\Theta}, W_1-W_2\rangle=\langle \check{\Theta}, \pa'(W_1-W_2)\rangle=
\langle \check{\Theta}, V_1-V_2\rangle
\end{align*} where we used $\pa'(W_1-W_2)=V_1-V_2$.
\end{remark}

\subsection{Classical Liouville action for Fuchsian groups}\label{ss:classical}
Let $e^{\varphi(z)}|dz|^2$ be the hyperbolic metric on $\Omega=\Omega_1\sqcup \Omega_2$. The critical point of the Liouville action  along a conformal family of metrics appears at the hyperbolic metric. We call
\begin{align*}
S =S[\varphi]=&\frac{i}{2}\Bigl(\langle \omega, F_1-F_2\rangle-\langle\check{\theta}, L_1-L_2\rangle+\langle \check{u}, W_1-W_2\rangle\Bigr)
\end{align*}the classical Liouville action. Here,\begin{align*}
\omega = \omega[\varphi],\hspace{1cm} \check{\theta}=\check{\theta}[\varphi].
\end{align*}

Let us consider the case where $\Gamma$ is a Fuchsian group.
In \cite{2} we prove that if $\Gamma$ is a cocompact Fuchsian group, then the classical Liouville action is equal to $8\pi(2g-2)$, where $g$ is the genus of the compact Riemann surface $\Gamma\backslash\mathbb{U}$. In other words,
$$S=8\pi \chi(X)$$ when $\Gamma\backslash\mathbb{U}\simeq X$ is a compact Riemann surface.
It is natural to ask whether this is still true when $\Gamma\backslash\mathbb{U}\simeq X$ is a surface of type $(g, n; m_1, m_2, \ldots, m_r)$.

When $\Gamma$ is a Fuchsian group,
$$\omega=2e^{\varphi}dz\wedge d\bar{z},$$so that $\delta\omega=0$. Hence $d\check{\theta}=0$. This implies that $\check{\theta}=d\varkappa$ for some $\varkappa$. It follows that
\begin{align*}
\langle \check{u}, W_1-W_2\rangle =\langle \delta d\varkappa, W_1-W_2\rangle=\langle \delta\varkappa, V_1-V_2\rangle=\langle \varkappa, \pa''(V_1-V_2)\rangle.
\end{align*}
Applying the equality of \eqref{eq0914_1} to the double complex associated to the action of the Fuchsian group on $\mathbb{U}\sqcup \mathbb{L}$, we find that
\begin{align}\label{eq:0914qF}
\pa'L_1-\pa'L_2=\ &\pa''V_1-\sum_{i=1}^nv_i\otimes [\kappa_i]-\sum_{j=1}^r z_{j}\otimes [\tau_j] \notag\\
&-\pa''V_2+\sum_{i=1}^nv_i\otimes [\kappa_i]+\sum_{j=1}^r\bar{z}_{j}\otimes [\tau_j]\notag \\
=\ &\pa''V_1-\pa''V_2-\sum_{j=1}^r \left(z_{j}-\bar{z}_{j}\right)\otimes [\tau_j].
\end{align}
Hence,
\begin{align*}
\langle \check{u}, W_1-W_2\rangle =\langle \check{\theta}, L_1-L_2\rangle +\sum_{j=1}^r\langle \varkappa, \left(z_j-\bar{z}_j\right)\otimes [\tau_j]\rangle.
\end{align*}Therefore, when $\Gamma$ is a Fuchsian group, the classical Liouville action is given by
\begin{align}\label{eq0915_2}
S=2\iint_{X\sqcup \bar{X}}e^{\varphi}d^2z+\frac{i}{2}\sum_{j=1}^r\langle \varkappa, \left(z_j-\bar{z}_j\right)\otimes [\tau_j]\rangle.
\end{align}
This shows that when $\Gamma$ does not contain elliptic elements, we indeed have
$S=8\pi \chi(X).$ When $\Gamma$ contains elliptic elements, there is an additional term given by
$$\frac{i}{2}\sum_{j=1}^r\langle \varkappa, \left(z_j-\bar{z}_j\right)\otimes [\tau_j]\rangle.$$Let us compute this term. By the definition of $\varkappa$, we have
\begin{align*}
 \langle \varkappa,  \left(z_j-\bar{z}_j\right)\otimes [\tau_j]\rangle=&\varkappa_{\tau_j}(z_j)-\varkappa_{\tau_j}(\bar{z}_j)\\
 =&\int_{\bar{z}_j}^{z_j}\check{\theta}_{\tau_j}\\
 =&-\int_{\bar{z}_j}^{z_j} \left(\log y^2
 +\log\left|\frac{\tau_j^{-1\prime\prime}}{\tau_j^{-1\prime}}\right|^2\right) \left(\frac{\tau_j^{-1\prime\prime}}{\tau_j^{-1\prime}}dz -
\frac{\ov{\tau_j^{-1\prime\prime}}}{\ov{\tau_j^{-1\prime}}}d\bar{z}\right).
\end{align*}
Now $\tau_j=\rho_j \lambda_{m_j} \rho_j^{-1}$, where
\begin{align}\label{e:def-rho-lambda2}
\rho_j=\begin{pmatrix}-\varpi e^{-i\alpha}\bar{z}_j & \varpi e^{i\alpha}z_j\\ -\varpi e^{-i\alpha}&\varpi e^{i\alpha}  \end{pmatrix}, \qquad
\lambda_{m_j}= \begin{pmatrix}e^{\frac{\pi i}{m_j}} & 0\\ 0 & e^{-\frac{\pi i}{m_j}} \end{pmatrix},
\end{align}where $\alpha$ is a constant and $\varpi^2=(z_j-\bar{z}_j)^{-1}$. Hence,
\begin{align*}
\tau_j^{-1}(z)= \frac{  (z_j e^{i\beta} -\bar{z}_j e^{-i\beta} )z -z_j\bar{z_j} (e^{i\beta} - e^{-i\beta}) }
{ (e^{i\beta}- e^{-i\beta})z-  (\bar{z}_je^{i\beta}- z_j e^{-i\beta})},
\end{align*}
where $\beta= \pi/m_j$. It follows that
\begin{equation*}
\frac{\tau_j^{-1\prime\prime}(z)}{\tau_j^{-1\prime}(z)}=-\frac{2 (e^{i\beta}- e^{-i\beta})}{ (e^{i\beta}- e^{-i\beta})z-  (\bar{z}_je^{i\beta}- z_j e^{-i\beta})}.
\end{equation*}We can choose the integration path from $\bar{z}_j$ to $z_j$ to be the straight line from $\bar{z}_j$ to $z_j$. Using the parametrization
$$s\in[-1,1]\quad \longrightarrow \quad z(s)=\frac{z_j+\bar{z}_j +s(z_j-\bar{z}_j)}{2},$$ we find that
\begin{align*}
&\langle \varkappa,  \left(z_j-\bar{z}_j\right)\otimes [\tau_j]\rangle\\=&2\int_{-1}^1 \log\frac{s^2(4\sin^2\beta)}{\left|\cos\beta+is\sin\beta\right|^2}d\log \frac{\cos\beta+is\sin\beta}{\cos\beta-is\sin\beta}\\
=&4\int_{0}^1 \log\frac{s^2(4\sin^2\beta)}{\left|\cos\beta+is\sin\beta\right|^2}d\log \frac{\cos\beta+is\sin\beta}{\cos\beta-is\sin\beta}\\
=&8\int_0^1\log s\; d\log \frac{\cos\beta+is\sin\beta}{\cos\beta-is\sin\beta}+4\log\left(4\sin^2\beta\right)\int_0^1d\log \frac{\cos\beta+is\sin\beta}{\cos\beta-is\sin\beta}\\&-4\int_0^1\log \left|\cos\beta+is\sin\beta\right|^2d\log \frac{\cos\beta+is\sin\beta}{\cos\beta-is\sin\beta}\\
=&(\text{I})+(\text{II})+(\text{III}).
\end{align*}
Let us recall the dilogarithm function
\begin{equation}\label{0915_1}\begin{split}
\mathrm{Li}_2(z):=&- \int_0^z \frac{\log (1-t)}{t} \, dt \\
=&-\int_0^1\frac{\log(1-zt)}{t}dt\\
=&\int_0^1\log t \;\;d\log(1-zt).
\end{split}
\end{equation}
This gives
\begin{align*}
(\text{I})=&8\int_0^1 \log s\; d\log \frac{1+is\tan\beta}{1-is\tan\beta}\\
=&8\text{Li}_2\left(-i\tan\beta\right)-8\text{Li}_2\left(i\tan\beta\right).
\end{align*}
It is straightforward to find that
\begin{align*}
(\text{II})=&4\log\left(4\sin^2\beta\right)\int_0^1d\log \frac{\cos\beta+is\sin\beta}{\cos\beta-is\sin\beta}\\
=&8i\beta \log\left(4\sin^2\beta\right).
\end{align*}
The computation of (\text{III}) is more complicated. It is given by the lemma below.

\begin{lemma}\label{lemma_n1}
Let $x$ and $y$ be constants with $x>0$ and $x^2+y^2=1$. Then
\begin{align*}
&\int_0^1\log|x+isy|^2d\log\frac{x+isy}{x-isy}\\=&\log(2x)\log\frac{x+iy}{x-iy}-\text{Li}_2\left(\frac{x+iy}{2x}\right)+\text{Li}_2\left(\frac{x-iy}{2x}\right).
\end{align*}
\end{lemma}
\begin{proof}
Let $u=x+isy$, $v=x-isy$. Then
\begin{align*}
I=&\int_0^1\log|x+isy|^2\;d\log\frac{x+isy}{x-isy}\\
=&\int_0^1\left(\log u\;d\log u+\log v\;d\log u-\log u\;d\log v-\log v\;d\log v\right).
\end{align*}On the other hand, integration by parts gives
\begin{align*}
I=&-\int_0^1 \log\frac{x+isy}{x-isy}\;d\log|x+isy|^2\\
=&-\int_0^1\left(\log u\;d\log u+\log u\;d\log v-\log v\;d\log u-\log v\;d\log v\right).
\end{align*}
It follows that
\begin{align*}
\int_0^1\left(\log u\;d\log u-\log v\;d\log v\right)=0
\end{align*}
and
\begin{align*}
I=&\int_0^1\left( \log v\;d\log u-\log u\;d\log v \right)\\
=&\int_0^1\left(\log(x-isy)\;d\log(x+isy)-\log(x+isy)\;d\log(x-isy)\right).
\end{align*}
Now
\begin{align*}
\int_0^1\log(x-isy)d\log (x+isy)=&\int_x^{x+iy}\log(2x-t)\frac{dt}{t}\\
=&\int_{\frac{1}{2}}^{\frac{x+iy}{2x}}\left(\log(2x)+\log (1-t)\right)\frac{dt}{t}\\
=&\log(2x)\log\frac{x+iy}{x}-\text{Li}_2\left(\frac{x+iy}{2x}\right)+\text{Li}_2\left(\frac{1}{2}\right).
\end{align*}
Changing $y$ to $-y$ gives
\begin{align*}
\int_0^1\log(x+isy)d\log (x-isy)
=&\log(2x)\log\frac{x-iy}{x}-\text{Li}_2\left(\frac{x-iy}{2x}\right)+\text{Li}_2\left(\frac{1}{2}\right).
\end{align*}
It follows that
\begin{align*}
&\int_0^1\log|x+isy|^2d\log\frac{x+isy}{x-isy}\\=&\log(2x)\log\frac{x+iy}{x-iy}-\text{Li}_2\left(\frac{x+iy}{2x}\right)+\text{Li}_2\left(\frac{x-iy}{2x}\right).
\end{align*}
\end{proof}

From Lemma \ref{lemma_n1}, we find that
\begin{align*}
(\text{III})=&-4\int_0^1\log \left|\cos\beta+is\sin\beta\right|^2d\log \frac{\cos\beta+is\sin\beta}{\cos\beta-is\sin\beta}\\
=&-8i\beta\log\left(2\cos\beta\right) +4\text{Li}_2\left(\frac{1}{1+e^{-2i\beta}}\right)-4\text{Li}_2\left(\frac{1}{1+e^{2i\beta}}\right).
\end{align*}
Hence,
\begin{align*}
&\langle \varkappa,  \left(z_j-\bar{z}_j\right)\otimes [\tau_j]\rangle\\=&
8\text{Li}_2\left(-i\tan\beta\right)-8\text{Li}_2\left(i\tan\beta\right)+8i\beta \log\left(4\sin^2\beta\right)\\
&-8i\beta\log\left(2\cos\beta\right) +4\text{Li}_2\left(\frac{1}{1+e^{-2i\beta}}\right)-4\text{Li}_2\left(\frac{1}{1+e^{2i\beta}}\right)\\
=&-16i\;\mathrm{Im} \text{Li}_2\left(i\tan\beta\right)+8i \;\mathrm{Im} \text{Li}_2\left(\frac{1}{1+e^{-2i\beta}}\right)+8i\beta\log\frac{2\sin^2\beta}{\cos\beta}.
\end{align*}
Now let us recall that the Bloch-Wigner function $D(z)$  is given by \cite{Zagier}:
\begin{equation}\label{e:def-BW}
D(z)= \mathrm{Im} \left( \mathrm{Li}_2(z) \right) +\mathrm{arg}(1-z) \log|z| \qquad \text{for} \quad z\in \mathbb{C}.
\end{equation}
Since
\begin{gather*}
\log|i\tan\beta|=\log\frac{\sin\beta}{\cos\beta},\hspace{1cm}\log\left|\frac{1}{1+e^{-2i\beta}}\right|=-\log(2\cos\beta),\\
\mathrm{arg}(1-i\tan\beta)=-\beta,\hspace{1cm}\mathrm{arg}\left(1-\frac{1}{1+e^{-2i\beta}}\right)=-\beta,
\end{gather*}
we find that
\begin{equation}\label{eq0915_3}\begin{split}
&\langle \varkappa,  \left(z_j-\bar{z}_j\right)\otimes [\tau_j]\rangle\\=
&-16i \;\left(D\left(i\tan\beta\right)+\beta\log\frac{\sin\beta}{\cos\beta}\right)+8i \;\left(D\left(\frac{1}{1+e^{-2i\beta}}\right)-\beta\log(2\cos\beta)\right)\\
&+8i\beta\log\frac{2\sin^2\beta}{\cos\beta}\\
=&-16 i\,D\left(i\tan\beta\right)+8i\,D\left(\frac{1}{1+e^{-2i\beta}}\right).\end{split}
\end{equation}
Using the identity (see \cite{Zagier})
\begin{align*}
D(z)=\frac{1}{2}\left[D\left(\frac{z}{\bar{z}}\right)+D\left(\frac{1-1/z}{1-1/\bar{z}}\right)+D\left(\frac{1/(1-z)}{1/(1-\bar{z})}\right)\right],
\end{align*}
we find that
\begin{align*}
D(i\tan\beta)=\frac{1}{2}\left[D(-1)+D(-e^{-2i\beta})+D(e^{2i\beta})\right].
\end{align*}By definition,
\begin{align*}
D\left(e^{i\theta}\right)=\sum_{n=1}^{\infty}\frac{\sin n\theta}{n^2}.
\end{align*}Hence, $D(-1)=0$. On the other hand, we also have the identity (see \cite{Zagier})
\begin{align*}
D(z)=D\left(1-\frac{1}{z}\right).
\end{align*}
Hence,
\begin{align*}
D\left(\frac{1}{1+e^{-2i\beta}}\right)=D(-e^{-2i\beta}).
\end{align*}
Therefore,
\begin{align*}
&\langle \varkappa,  \left(z_j-\bar{z}_j\right)\otimes [\tau_j]\rangle=-8i\,D(e^{2i\beta}).
\end{align*}

Gathering the results above, we have
\begin{theorem}\label{classical_Liouville}
When $\Gamma$ is a Fuchsian group of type $(g; n; m_1, m_2, \ldots, m_r)$, the classical Liouville action $S$ is given by
\begin{align*}
S=8\pi \chi(X)+4\sum_{j=1}^r  D\left( e^{\frac{2\pi i}{m_j}} \right),
\end{align*} where $\chi(X)$ is given by \eqref{e:def-chi}.
\end{theorem}

\section{Variations of the classical Liouville action}\label{s:variation}

In this section, we want to compute the first and second variations of the classical Liouville action on $\mathfrak{D}(\Gamma)$. Most of the computations are similar to the one given in \cite{2}  when $\Gamma$ does not contain parabolic or elliptic elements. However, we have to be careful when analysing the possible singularities at the parabolic and elliptic fixed points. There might also be extra terms appearing at the elliptic fixed points since the double complex $\Sigma_1-\Sigma_2$ is not a cycle.

Given a harmonic Beltrami differential $\mu$,
let $f^{\mu}$ be the unique quasi-conformal mapping with Beltrami differential $\mu$ that fixes the points $0, 1, \infty$.
Notice that $f^{\vep\mu}$ varies holomorphically with respect to $\vep$ and thus
\begin{align*}
\left.\frac{\pa}{\pa\bar{\vep}}f^{\vep\mu}\right|_{\vep=0}=0.
\end{align*}
Let
\begin{align*}
\dot{f}=\left.\frac{\pa}{\pa\vep}f^{\vep\mu}\right|_{\vep=0}.
\end{align*}
It follows from the definition $$f^{\vep\mu}_{\bar{z}}=\vep\mu f^{\vep\mu}_z$$ that
$$\dot{f}_{\bar{z}}=\mu.$$
For any linear fractional transformation $\gamma$, let $$\gamma^{\vep\mu}=f^{\vep\mu}\circ\gamma\circ (f^{\vep\mu})^{-1}.$$ Then $\gamma^{\vep\mu}$ varies holomorphically with respect to $\vep$ .
We collect some formulas in the following theorem.

\begin{lemma}\label{theorem1}
Let $\mu$ be a harmonic Beltrami differential of $\Gamma$. On $\mathfrak{D}(\Gamma)$, we have the following variation formulas:
\begin{enumerate}
\item[(i)] $\displaystyle \left.\frac{\pa}{\pa\vep}\left(e^{\varphi^{\vep\mu}\circ f^{\vep\mu}}\left|f_z^{\vep\mu}\right|^2\right)\right|_{\vep=0}=0$,
\item[(ii)] $\displaystyle \left.\frac{\pa}{\pa\vep}\left(\varphi^{\vep\mu}\circ f^{\vep\mu}\right)\right|_{\vep=0}=-\dot{f}_z$,

\item[(iii)] $\displaystyle \left.\frac{\pa}{\pa\vep}\left(\varphi^{\vep\mu}_z\circ f^{\vep\mu}f_z^{\vep\mu}\right)\right|_{\vep=0}=-\dot{f}_{zz}$,

\item[(iv)] $\displaystyle \left.\frac{\pa}{\pa\vep}\left(\varphi^{\vep\mu}_{\bar{z}}\circ f^{\vep\mu}\bar{f}_{\bar{z}}^{\vep\mu}\right)\right|_{\vep=0}=0$,

\item[(v)] $\displaystyle \left.\frac{\pa}{\pa\vep}\left(\log\left|\left(\gamma^{\vep\mu}\right)'\circ f^{\vep\mu}\right|^2\right)\right|_{\vep=0}=\dot{f}_z\circ\gamma-\dot{f}_z$,

\item[(vi)] $\displaystyle \left.\frac{\pa}{\pa\vep}\left(\frac{\left(\gamma^{\vep\mu}\right)^{\prime\prime}}{\left(\gamma^{\vep\mu}\right)^{\prime}}\circ f^{\vep\mu}f^{\vep\mu}_z \right)\right|_{\vep=0}=\dot{f}_{zz}\circ\gamma \gamma'-\dot{f}_{zz}$,

    \item[(vii)] $\displaystyle \left.\frac{\pa}{\pa\vep}\left(\frac{\left(\gamma^{\vep\mu}\right)^{\prime\prime}}{\left(\gamma^{\vep\mu}\right)^{\prime}}\circ f^{\vep\mu} \right)\right|_{\vep=0}=\dot{f}_{zz}\circ\gamma \gamma'-\dot{f}_{zz}-\frac{\gamma''}{\gamma'}\dot{f}_z$.
\end{enumerate}

\end{lemma}

\begin{proof}
The proof of (i) follows from the classical result of Ahlfors \cite{3}. The equalities (ii), (iii), (v) and (vi) are given in p.213-214 of \cite{2}. For (iv), we start with the following equality given in (3.6) of \cite{2}:
\begin{align*}
\left.\frac{\pa}{\pa\vep}\left(\varphi^{\vep\mu}_{\bar{z}}\circ f^{\vep\mu}\bar{f}_{\bar{z}}^{\vep\mu}\right)\right|_{\vep=0}=&-\left(\varphi_z\dot{f}_{\bar{z}}+\dot{f}_{z\bar{z}}\right)\\=&
-\left(\varphi_z\mu+\mu_z\right).
\end{align*}Since $\mu$ is a harmonic Beltrami differential, i.e.,
$$\mu=e^{-\varphi} \bar{q}$$ for some holomorphic quadratic differential $q$, it follows immediately that
\begin{align}\label{eq11_14_5}\varphi_z\mu+\mu_z=0,\end{align}and (iv) is proved. (vii) follows immediately from (vi).
\end{proof}

\vspace{0.5cm}
\begin{lemma}\label{lemma1}We have the following formulas:
\begin{align}\label{eq11_14_2}
\dot{f}_{zz}\circ\gamma\gamma'-\dot{f}_{zz}=\frac{1}{2}\left(\dot{f}_z\circ\gamma+\dot{f}_z\right)\frac{\gamma''}{\gamma'}+\frac{\dot{c}(\gamma)}{c(\gamma)}\frac{\gamma''}{\gamma'} \, ,
\end{align}
\begin{align}\label{eq11_14_4}
\dot{f}_{z\bar{z}}\circ\gamma\overline{\gamma'}-\dot{f}_{z\bar{z}}=\dot{f}_{\bar{z}}\frac{\gamma''}{\gamma'} \, .
\end{align}
\end{lemma}
\begin{proof}
Here we give a proof for \eqref{eq11_14_2} that is much simpler than the one given in \cite{2}. From \eqref{eq11_14_1}, we have
\begin{align*}
\log|2c(\gamma)|^2=\log\left|\frac{\gamma''}{\gamma'}\right|^2-\frac{1}{2}\log|\gamma'|^2.
\end{align*}
It follows from (vii) and (v) of Theorem \ref{theorem1} that
\begin{align*}
\frac{\dot{c}(\gamma)}{c(\gamma)}=\frac{\dot{f}_{zz}\circ\gamma\gamma'-\dot{f}_{zz}-\displaystyle\frac{\gamma''}{\gamma'}\dot{f}_z}{\displaystyle\frac{\gamma''}{\gamma'}}-\frac{1}{2}
\left(\dot{f}_z\circ\gamma-\dot{f}_z\right),
\end{align*}from which \eqref{eq11_14_2} follows immediately.
The equation \eqref{eq11_14_4} follows immediately by differentiating the identity
\begin{align*}
\dot{f}_{\bar{z}}\circ\gamma\frac{\overline{\gamma'}}{\gamma'}=\dot{f}_{\bar{z}}
\end{align*}with respect to $z$.
\end{proof}

The Lie derivatives of the smooth family  of $(l, m)$ tensors $\omega$ on $\mathfrak{D}(\Gamma)$ along the vector field  determined by $\mu$   are defined as
\begin{align*}
L_{\mu}\omega=&\left.\frac{\pa}{\pa\vep}\right|_{\vep=0}\omega^{\vep\mu}\circ f^{\vep\mu} (f^{\vep\mu}_z)^l\left(\overline{f^{\vep\mu}_z}\right)^m,\\
L_{\bar{\mu}}\omega=&\left.\frac{\pa}{\pa\bar{\vep}}\right|_{\vep=0}\omega^{\vep\mu}\circ f^{\vep\mu} (f^{\vep\mu}_z)^l\left(\overline{f^{\vep\mu}_z}\right)^m.
\end{align*}
Let
\begin{align*}
\vartheta(z)=&2\varphi_{zz}-\varphi_z^2\\
=&\begin{cases} 2\mathcal{S}\left(J_1^{-1}\right)(z),\quad&\text{if}\quad z\in \Omega_1\\
2\mathcal{S}\left(J_2^{-1}\right)(z),\quad&\text{if}\quad z\in \Omega_2\end{cases},
\end{align*}
where
\begin{align*}
\mathcal{S}(h)=&\left(\frac{h_{zz}}{h_z}\right)_z-\frac{1}{2}\left(\frac{h_{zz}}{h_z}\right)^2\\
=& \frac{h_{zzz}}{h_z}-\frac{3}{2}\left(\frac{h_{zz}}{h_z}\right)^2
\end{align*}is the Schwarzian derivative of $h$.

\vspace{0.5cm}
\begin{theorem}\label{firstvariation}On the deformation space $\mathfrak{D}(\Gamma)$,
\begin{align*}
L_{\mu} S=\iint\limits_{\Gamma\backslash \Omega}\vartheta(z)\mu(z)d^2z.
\end{align*}Equivalently,
\begin{align*}
\pa S=\vartheta.
\end{align*}
\end{theorem}

\begin{proof}

By definition,
\begin{align*}
L_{\mu}S =&\frac{i}{2}\Bigl(\langle L_{\mu}\omega, F_1-F_2\rangle-\langle L_{\mu}\check{\theta}, L_1-L_2\rangle+\langle L_{\mu}\check{u}, W_1-W_2\rangle\Bigr).
\end{align*}
As in p. 217 of \cite{2}, we find that
\begin{align*}
L_{\mu}\omega = &- \varphi_{\bar{z}}\dot{f}_{zz} dz\wedge d\bar{z}\\
=&\vartheta \mu dz\wedge d\bar{z}-d\xi,
\end{align*}where
\begin{align*}
\xi=&2\varphi_z\dot{f}_{\bar{z}}d\bar{z}-\varphi d \dot{f}_z=-2\dot{f}_{z\bar{z}}d\bar{z}-\varphi d \dot{f}_z.
\end{align*}
Here we have used the equality \eqref{eq11_14_5}.
It follows that
\begin{align*}
\langle L_{\mu}\omega, F_1-F_2\rangle =& \langle \vartheta \mu dz\wedge d\bar{z}, F_1-F_2\rangle -\langle d\xi, F_1-F_2\rangle\\
=&\langle \vartheta \mu dz\wedge d\bar{z}, F_1-F_2\rangle -\langle \delta\xi, L_1-L_2\rangle.
\end{align*} 
The second equality follows from Theorem \ref{t:justification-Stokes} and
$$\partial' (F_1-F_2)=\partial'' (L_1-L_2).$$
Now as in p.218 of \cite{2}, using the equality \eqref{eq11_14_4} we obtain
\begin{align*}
\delta\xi_{\gamma^{-1}}=-2\left(\dot{f}_{z\bar{z}}\circ\gamma\overline{\gamma'}-\dot{f}_{z\bar{z}}\right)d\bar{z}-\varphi d\left(\dot{f}_z\circ\gamma-\dot{f}_z\right)+\log|\gamma'|^2 d\left(\dot{f}_z\circ\gamma\right).
\end{align*}
Now, for the second term $\langle L_{\mu}\check{\theta}, L_1-L_2\rangle$, by (v) of Lemma \ref{theorem1},
\begin{align*}
&L_{\mu}\left(\varphi-\frac{1}{2}\log|\gamma'|^2-2\log2-\log|c(\gamma)|^2\right)\\
=&-\dot{f}_z-\frac{1}{2}\left(\dot{f}_z\circ\gamma-\dot{f}_z\right)
-\frac{\dot{c}(\gamma)}{c(\gamma)}
=-\frac{1}{2}\left(\dot{f}_z\circ\gamma+\dot{f}_z\right)-\frac{\dot{c}(\gamma)}{c(\gamma)}\\
&\hspace{5.3cm} =-\frac{\dot{f}_{zz}\circ\gamma\gamma'-\dot{f}_{zz}}{\displaystyle\frac{\gamma''}{\gamma'}}
\end{align*}
where the last line follows from \eqref{eq11_14_2} and
\begin{align*}
L_{\mu}\left(\frac{\gamma''}{\gamma'}dz-\frac{\overline{\gamma''}}{\overline{\gamma'}}d\bar{z}\right)=&L_{\mu}d\log|\gamma'|^2
=d\left(\dot{f}_z\circ\gamma-\dot{f}_z\right).
\end{align*}
Hence,
\begin{align*}
L_{\mu}\check{\theta}_{\gamma^{-1}}=&\left(\varphi-\frac{1}{2}\log|\gamma'|^2-2\log2-\log|c(\gamma)|^2\right)d\left(\dot{f}_z\circ\gamma-\dot{f}_z\right)\\
&-\left( \frac{1}{2}\left(\dot{f}_z\circ\gamma+\dot{f}_z\right)+\frac{\dot{c}(\gamma)}{c(\gamma)}\right)\left(2\frac{\gamma''}{\gamma'}dz-d\log|\gamma'|^2\right)\\
=&\left(\varphi-\frac{1}{2}\log|\gamma'|^2-2\log2-\log|c(\gamma)|^2\right)d\left(\dot{f}_z\circ\gamma-\dot{f}_z\right)\\&-2\left(\dot{f}_{zz}\circ\gamma\gamma'-\dot{f}_{zz}\right)dz
+\left( \frac{1}{2}\left(\dot{f}_z\circ\gamma+\dot{f}_z\right)+\frac{\dot{c}(\gamma)}{c(\gamma)}\right)d\log|\gamma'|^2.
\end{align*}
Let
$$\chi=\delta\xi+L_{\mu}\check{\theta}.$$It follows that
\begin{align*}
\chi_{\gamma^{-1}}=& dl_{\gamma^{-1}},
\end{align*}
where
\begin{align*}
l_{\gamma^{-1}}=&\frac{1}{2}\log|\gamma'|^2\left(\dot{f}_z\circ\gamma +\dot{f}_z+2\frac{\dot{c}(\gamma)}{c(\gamma)}\right)\\
&-\left(\log|c(\gamma)|^2+2+2\log 2\right)\left(\dot{f}_z\circ\gamma-\dot{f}_z\right).
\end{align*}By Lemma \ref{lemma3} and Lemma \ref{lemma5}, $l_{\kappa_i}$ is well-defined when $z$ approaches the fixed point of $\kappa_i$ on $\mathcal{C}$.
Hence, we have
\begin{align*}
\langle \delta\xi+L_{\mu}\check{\theta}, L_1-L_2\rangle =&\langle dl, L_1-L_2\rangle \\
=&\langle l, \pa'L_1-\pa'L_2\rangle.
\end{align*}
Since
\begin{align*}
L_{\mu}\check{u}=&L_{\mu}\delta \check{\theta}=\delta L_{\mu}\check{\theta}
=\delta \chi
=\delta dl
=d\delta l,
\end{align*}
we have
\begin{align*}
\left\langle L_{\mu}\check{u}, W_1-W_2\right\rangle =&\left\langle\delta l, \pa'(W_1-W_2)\right\rangle\\
=&\left\langle\delta l, V_1-V_2\right\rangle\\
=&\left\langle l, \pa''V_1-\pa''V_2\right\rangle.
\end{align*}
From this it follows that
\begin{align*}
L_{\mu}S= \frac{i}{2}\Bigl(\langle \vartheta \mu dz\wedge d\bar{z}, F_1-F_2\rangle-\langle  l, \pa'L_1-\pa'L_2-\pa''V_1+\pa''V_2\rangle\Bigr).
\end{align*}
Let $v_i\in\mathcal{C}$ be the fixed point of the parabolic generator $\kappa_i$,  and let $\mathrm{w}_{1j}$ and $\mathrm{w}_{2j}$ be the fixed points of the elliptic generator $\tau_j$ in $\Omega_1$ and $\Omega_2$ respectively. Then
\begin{align*}
\pa'L_1-\pa'L_2=&\pa''V_1-\sum_{i=1}^nv_i\otimes [\kappa_i]-\sum_{j=1}^r \mathrm{w}_{1j}\otimes [\tau_j]\\
&-\pa''V_2+\sum_{i=1}^nv_i\otimes [\kappa_i]+\sum_{j=1}^r\mathrm{w}_{2j}\otimes [\tau_j]\\
=&\pa''V_1-\pa''V_2-\sum_{j=1}^r \left(\mathrm{w}_{1j}-\mathrm{w}_{2j}\right)\otimes [\tau_j].
\end{align*}
Hence,
\begin{align*}
L_{\mu}S= & \frac{i}{2}\left(\langle \vartheta \mu dz\wedge d\bar{z}, F_1-F_2\rangle+\left\langle l, \sum_{j=1}^r (\mathrm{w}_{1j}-\mathrm{w}_{2j})\otimes [\tau_j]\right\rangle\right).
\end{align*}
Using Lemma \ref{lemma1} and the fact that $\tau_j(\mathrm{w}_{1j})=\mathrm{w}_{1j}$, $\tau_j(\mathrm{w}_{2j})=\mathrm{w}_{2j}$, one finds that $$\left\langle l, \sum_{j=1}^r (\mathrm{w}_{ij}-\mathrm{w}_{2j})\otimes [\tau_j]\right\rangle=0.$$  Hence, the presence of elliptic fixed points does not contribute additional terms.
This concludes that
\begin{align*}
L_{\mu}S= &\frac{i}{2}\langle \vartheta \mu dz\wedge d\bar{z}, F_1-F_2\rangle,
\end{align*}which is the assertion of the theorem.

 \end{proof}

 Before going to the second variation, let us collect some additional variation formulas.
 \begin{lemma}\label{theorem2}Let $\mu$ be a harmonic Beltrami differential of $\Gamma$. On $\mathfrak{D}(\Gamma)$, we have the following variation formulas:
 \begin{enumerate}
\item[(i)] $\displaystyle \left.\frac{\pa}{\pa\bar{\vep}}\left(\varphi^{\vep\mu}\circ f^{\vep\mu}\right)\right|_{\vep=0} =-\bar{\dot{f}}_{\bar{z}}$,
\item[(ii)] $\displaystyle \left.\frac{\pa}{\pa\bar{\vep}}\left(\varphi^{\vep\mu}_z\circ f^{\vep\mu}f^{\vep\mu}_z\right)\right|_{\vep=0} =0 $,

 \item[(iii)] $\displaystyle \left.\frac{\pa}{\pa\bar{\vep}}\left(\varphi^{\vep\mu}_{zz}\circ f^{\vep\mu}\left(f^{\vep\mu}_z\right)^2\right)\right|_{\vep=0} =-\frac{1}{2}e^{\varphi} \bar{\mu}$.

\end{enumerate}

 \end{lemma}

 \begin{proof}
 (i) is the complex conjugate of the formula (ii) in Lemma \ref{theorem1}. Differentiate this formula with respect to $z$, we have
 \begin{align*}
 \left.\frac{\pa}{\pa\bar{\vep}}\left(\varphi^{\vep\mu}_z\circ f^{\vep\mu}f^{\vep\mu}_z+\varphi^{\vep\mu}_{\bar{z}}\circ f^{\vep\mu}\overline{f^{\vep\mu}}_z\right)\right|_{\vep=0}
 =-\bar{\dot{f}}_{z\bar{z}}.
 \end{align*}Using $\bar{\dot{f}}_{z}=\bar{\mu}$, we have
 \begin{align}\label{eq11_14_6}
 \left.\frac{\pa}{\pa\bar{\vep}}\left(\varphi^{\vep\mu}_z\circ f^{\vep\mu}f^{\vep\mu}_z\right)\right|_{\vep=0} =-\bar{\mu}_{\bar{z}}-\varphi_{\bar{z}}\bar{\mu}=0.
 \end{align}The last equality follows from \eqref{eq11_14_5}. Differentiating \eqref{eq11_14_6} again with respect to $z$, we have
 \begin{align*}
 \left.\frac{\pa}{\pa\bar{\vep}}\left(\varphi^{\vep\mu}_{zz}\circ f^{\vep\mu}\left(f^{\vep\mu}_z\right)^2+\varphi^{\vep\mu}_{z\bar{z}}\circ f^{\vep\mu}f^{\vep\mu}_z\overline{f^{\vep\mu}}_z+\varphi^{\vep\mu}_z\circ f^{\vep\mu}f^{\vep\mu}_{z{z}}\right)\right|_{\vep=0}
 =0.
 \end{align*}
 This gives
 \begin{align*}
 \left.\frac{\pa}{\pa\bar{\vep}}\left(\varphi^{\vep\mu}_{zz}\circ f^{\vep\mu}\left(f^{\vep\mu}_z\right)^2\right)\right|_{\vep=0} =-\varphi_{z\bar{z}}\bar{\mu}.
 \end{align*}
 Then (iii) follows from the fact that  the hyperbolic metric $e^{\varphi(z)}|dz|^2$ satisfies the Liouville equation
  $$\varphi_{z\bar{z}}=\frac{1}{2}e^{-\varphi}.$$
 \end{proof}

\begin{theorem}\label{thm:second-var}
Let $\omega_{WP}$ be the symplectic form of the Weil-Petersson metric on $\mathfrak{D}(\Gamma)$. On $\mathfrak{D}(\Gamma)$,
\begin{align*}
d\vartheta=\bar{\pa}\pa S=-2i\omega_{WP}.
\end{align*}Hence, $-S$ is a K\"ahler potential of the Weil-Petersson metric on $\mathfrak{D}(\Gamma)$.
\end{theorem}

\begin{proof}
By the definition of $\langle \, , \, \rangle_{\text{WP}}$ in \eqref{e:def-WP pairing},
we need to show that
\begin{align*}
L_{\bar{\nu}}L_{\mu} S=-\iint\limits_{\Gamma\backslash\Omega}\mu(z)\overline{\nu(z)}\rho(z)d^2z.
\end{align*}
Using the result of Theorem \ref{firstvariation}, the fact that $L_{\bar{\nu}}\mu=0$, as well as (ii) and (iii) in Lemma \ref{theorem2}, we have
\begin{align*}
L_{\bar{\nu}}L_{\mu}S=&  \frac{i}{2}\left\langle L_{\bar{\nu}}\left(2\varphi_{zz}(z)-\varphi_z(z)^2\right) \mu dz\wedge d\bar{z}, F_1-F_2\right\rangle\\
=&-\iint\limits_{\Gamma\backslash\Omega } \mu(z)\overline{\nu(z)}\rho(z)d^2z.
\end{align*}The result follows.

\end{proof}

\section{Renormalized volume of quasi-Fuchsian 3-manifold}\label{s:holography}

The group $\text{PSL}(2, \mathbb{C})$ acts on the hyperbolic three space
$$\mathbb{U}^3=\{Z=(z,t) \,|\,z\in\mathbb{C}, \;t>0\}$$ and its closure. Given
$\displaystyle \gamma=\begin{pmatrix} a & b\\ c & d\end{pmatrix}\in \text{PSL}(2, \mathbb{C})$, let $$J_{\gamma}(Z)=\frac{1}{|cz+d|^2+|ct|^2}.$$Then
$\gamma$ maps $Z$ to $\gamma Z$, where
\begin{align*}
z(\gamma Z)=& \left((az+b)\overline{(cz+d)}+a\bar{c} t^2\right)J_{\gamma}(Z),\\
t(\gamma Z)=& tJ_{\gamma}(Z).
\end{align*}
Given a quasi-Fuchsian group $\Gamma\subset\text{PSL}(2, \mathbb{C})$, let $M=\Gamma\backslash\mathbb{U}^3$ be the quotient 3-manifold, which is called
\emph{quasi-Fuchsian 3-manifold}. The boundary of $M$ is $\pa'M=X\sqcup Y\simeq \Gamma\backslash\Omega$. In this section, we want to define the renormalized volume of $M$ and prove its relation to the Liouville action.

\subsection{Rank one cusps} Let $\zeta$ be  a parabolic fixed point of a quasi-Fuchsian group $\Gamma$ and $Stab_\zeta$ be the parabolic subgroup fixing  $\zeta$. We  call $\zeta$ a rank one or two cusp if $Stab_\zeta$ has one or two generators respectively. From now on we consider only the rank one cusp.

The quotient of the horoball $\mathcal{H}_s=\{(z,t)\in \mathbb{U}^3\, |\, t\geq s\}$ by the rank one parabolic subgroup $Stab_\zeta$ may or may not be embedded in $M=\Gamma\backslash \mathbb{U}^3$. Once it is embedded for some $t=s_0$, it is also embedded for all  larger values of  $s_0$. In this case, the embedded image is the same as $\pi(\mathcal{H}_{s_0})$, where $\pi:\mathbb{U}^3\to M$ denotes the projection map. This subset is homeomorphic to $\{ 0 < |z| \leq 1\}\times \mathbb{R}$.   We refer to this as a \emph{solid cusp tube} and its boundary $\pi(\partial \mathcal{H}_{s_0})$ as a \emph{cusp cylinder}.
A solid cusp tube has an infinite volume and a cusp cylinder has an infinite area.  For a sufficiently large $s$,
$\gamma (\mathcal{H}_s)\cap \mathcal{H}_s =\varnothing$  for  $\gamma\notin Stab_\zeta$, while $\gamma( \mathcal{H}_{s})=\mathcal{H}_s$ for $\gamma\in Stab_\zeta$.

A solid cusp tube is related to two punctures on the
boundary $X\sqcup Y$. There exists a pair of punctures $p_1$ on $X$,  $p_2$ on $Y$, uniquely associated with the conjugacy class of the rank one cusp.
If $c_1$ in $X$, $c_2$ in $Y$ are small circles retractible to $p_1$, $p_2$ respectively, there is a \emph{pairing cylinder} $C$ in $M$, which is a cylinder closed in $M$, and bounded by $c_1$, $c_2$. It bounds a subregion of $M$
called a \emph{solid pairing tube}, which is homeomorphic to $C\times(0,1]$. The solid pairing tubes corresponding to the different conjugacy classes of rank one cusps can be chosen to be mutually disjoint in $M$. The circles $c_1$, $c_2$
can be chosen so that the pair lifts to a round circles in $\Omega$ mutually tangent at the fixed point $\zeta$.
Such a pair of circles is called a \emph{double horocycle} at $\zeta$.

Let us consider a special case when $\Gamma$ is a Fuchsian group.
Suppose $\gamma_0:z\mapsto z+1$ is a generator of a rank one parabolic subgroup.
Then $\{ z\in\mathbb{C}\,|\, \mathrm{Im} (z)=\pm b\}$ for a constant $b>1$ is a double horocycle at the fixed point $\infty$.
Let $P_\pm\subset \mathbb{U}^3$ denote the vertical planes rising from them and consider $Q=\{ (z,t)\in\mathbb{U}^3\, |\, -b\leq \mathrm{Im} (z) \leq b, t>0\}$ they bound. Truncate $Q$ by the half space
$K=\{(z,t)\in \mathbb{U}^3 \, | \, t\geq a\}$ for a constant $a>1$. The relative boundary in $\mathbb{U}^3$ of the resulting tunnel $Q\setminus K$ projects to a pairing cylinder
in $M$. We refer to the section 3.6 of \cite{Mar} for more explanations about rank one cusps.

When the rank one cusp $v_i=\infty$ is associated to the parabolic subgroup generated by
$\kappa_i=\begin{pmatrix} 1 & q_i\\ 0 & 1 \end{pmatrix}$, we have
\begin{align}\label{e:sigma}
\sigma_i^{-1} \kappa_i\sigma_i=\begin{pmatrix} 1 & 1 \\ 0 & 1\end{pmatrix},\qquad \text{where}\quad
\sigma_i=\begin{pmatrix} q_i^{\frac12} & 0\\ 0 & q_i^{-\frac12} \end{pmatrix},
\end{align}
and $\sigma_i$ maps a horoball $\mathcal{H}_s$ onto $\mathcal{H}_{|q_i|s}$.
Hence, for $v_i=\infty$ with the associated parabolic element $\begin{pmatrix} 1 & q_i\\ 0 & 1 \end{pmatrix}$, we define
$\mathcal{H}_{i,\vep}$ to be $\mathcal{H}_{|q_i|/\vep}=\sigma_i(\mathcal{H}_{1/\vep})$.
When the rank one cusp $v_i$ is finite and associated with the parabolic subgroup generated by
\begin{align}\label{e:def-kappa}
\kappa_i=\begin{pmatrix}
1+q_iv_i & -q_i v_i^2\\ q_i & 1-q_iv_i \end{pmatrix},
\end{align}
we have
\begin{align}\label{e:sigma2}
\sigma_i^{-1}\kappa_i \sigma_i = \begin{pmatrix} 1 & -1 \\ 0 & 1 \end{pmatrix}, \qquad \text{where} \quad
\sigma_i=\begin{pmatrix} q_i^{\frac12}v_i & -q_i^{-\frac12}\\ q_i^{\frac12} & 0 \end{pmatrix}.
\end{align}
In this case,
\begin{lemma}\label{l:trans}
For the element $\sigma_i\in PSL(2,\mathbb{C})$ and an open horoball
$\mathcal{H}_s$, the image
$\sigma_i(\mathcal{H}_s)$ is an open horoball tangent to $\mathbb{C}$ at $\sigma_i(\infty)=v_i$ and with radius
$(2|q_i|s)^{-1}$.
\end{lemma}
\begin{proof}
It is well known that the image by $\sigma_i$ of an open horoball $\mathcal{H}_s$ is a open horoball tangent to $\mathbb{C}$. Hence, it is sufficient to find the tangency point of the horoball and its radius.
For $Z=(z,t)$, the image of $\sigma_i$ in \eqref{e:sigma2} of $Z$ is given by
\begin{align}\label{e:sigmafor}
\sigma_i(Z)=\left( v_i -\frac{1}{q_i}\, \frac{\bar{z}}{|z|^2+t^2}, \quad \frac{1}{|q_i|}\, \frac{t}{|z|^2+t^2} \right).
\end{align}
Using this, one can check that  for a fixed $x_0$ the image of a circle
$$\{(z,t)\in\mathbb{U}^3\, | z=x_0+iy, t=s\}\cup \{\infty\}$$
under $\sigma_i$ is
a circle satisfying $|Z-Z_0|^2= (2|q_i|s)^{-2}$ where $Z_0=(v_i, (2|q_i|s)^{-1})$.  Moreover it is easy to check that this circle is tangent to $\mathbb{C}$ at $v_i$. The claim follows from this.
\end{proof}
Using Lemma \ref{l:trans}, for finite $v_i$ we define $\mathcal{H}_{i,\vep}$ to be
the image by $\sigma_i$ of an open horoball $\mathcal{H}_{1/\vep}$,
which is an open horoball tangent to $\mathbb{C}$ at $\sigma_i(\infty)=v_i$ with  radius  ${\vep}/{2|q_i|}$.

For   $s\gg 0$, we consider $\mathcal{K}_{s}$  defined by
\begin{align*}
\mathcal{K}_s:=\{(z,t)\in\mathbb{U}^3\,|\, z=x+iy, |y| > s \}.
\end{align*}
Note that the boundary of $\mathcal{H}_{s}\cup\mathcal{K}_{s}$ for  $s\gg0$
projects to a pairing cylinder in $M$ if $\infty$ is a rank one cusp.
For $v_i=\infty$, we put $\mathcal{K}_{i,\vep}$ to be $\mathcal{K}_s$ with $s=\vep^{-1/2}$.
For other finite $v_i$, we define $\mathcal{K}_{i,\vep}$ to be the image by $\sigma_i$ of $\mathcal{K}_s$ with $s=\vep^{-1/2}$.

\subsection{Conical singularity}
 For an elliptic element $\tau_j$ of order $m_j$ with fixed points $\mathrm{w}_1$ and $\mathrm{w}_2$, we have that $\tau_j=\rho_j \lambda_{m_j} \rho_j^{-1}$, where
\begin{align}\label{e:def-rho-lambda}
\rho_j=\begin{pmatrix}-\varpi e^{-i\alpha}\mathrm{w}_2 & \varpi e^{i\alpha}\mathrm{w}_1\\ -\varpi e^{-i\alpha}&\varpi e^{i\alpha}  \end{pmatrix}, \qquad
\lambda_{m_j}= \begin{pmatrix}e^{\frac{\pi i}{m_j}} & 0\\ 0 & e^{-\frac{\pi i}{m_j}} \end{pmatrix}
\end{align}
for some $\alpha$ and
$\varpi^2=({\mathrm{w}_1-\mathrm{w}_2})^{-1}$. Note that the $t$-axis   in $\mathbb{U}^3$ is fixed by  $\lambda_{m_j}$  and it is mapped to the geodesic $h_j$ by $\rho_j$
which is a semicircle in $\mathbb{U}^3$
with two end points $\mathrm{w}_1$ and $\mathrm{w}_2$. We consider a neighborhood of the $t$-axis defined by
\begin{align*}
\mathcal{Q}_\vep=\{ (z,t)\in \mathbb{U}^3\, | \, t\geq \vep^{-1} |z| \}.
\end{align*}
Then the neighborhood $\mathcal{Q}_\vep$ is mapped to a neighborhood $\mathcal{Q}_{j,\vep}:=\rho_j(\mathcal{Q}_\vep)$ of the geodesic $h_j$. Note that this subset
$\mathcal{Q}_{j,\vep}$ is invariant under the action of $\tau_j$.

For a sufficiently small $\vep>0$, the quotient of $\mathcal{Q}_{j,\vep}$ by the finite group generated by $\tau_j$ can be embedded into
$M$, and the images corresponding to pairs of elliptic fixed points can be mutually disjoint in $M$. The hyperbolic metric over these regions
has the \emph{conical singularity} of codimension $2$ since the angle around the projection image of the geodesic $h_j$ is $2\pi/m_j$.

 \begin{figure}[h]
\epsfxsize=1\linewidth \epsffile{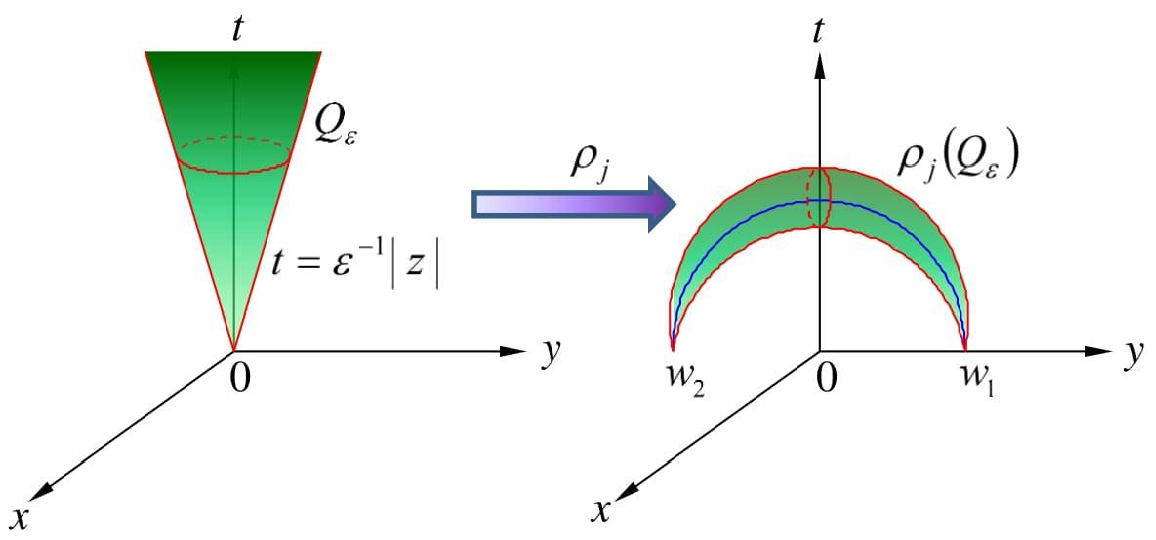}   \caption{\label{f2} $\mathcal{Q}_{\vep}$ and $\mathcal{Q}_{j,\vep}:=\rho_j(\mathcal{Q}_{\vep})$.}\end{figure}

A direct computation gives the following result.

\begin{lemma}\label{l:ell-geo} For sufficiently small $\vep>0$, two points $ (0, \vep)$ and $(0,\vep^{-1})$ in the $t$-axis are mapped by $\rho_j$ to two points in the geodesic $h_j$ whose coordinates are given by
\begin{align*}
\rho_j \left((0,\vep)\right)=&\left(\mathrm{w}_1 +(\mathrm{w}_2-\mathrm{w}_1)\vep^2 +O(\vep^4) ,\,  |\mathrm{w}_1-\mathrm{w}_2| (\vep - \vep^3)+ O(\vep^5)\right),\\
\rho_j \left((0,\vep^{-1})\right)=&\left(\mathrm{w}_2 +(\mathrm{w}_1-\mathrm{w}_2)\vep^2 +O(\vep^4) ,\,  |\mathrm{w}_1-\mathrm{w}_2| (\vep - \vep^3)+ O(\vep^5)\right).
\end{align*}
\end{lemma}

\vspace{0.5cm}
\subsection{Chain complex of the quasi-Fuchsian 3-manifold}
When $\Gamma$ is a Fuchsian group, we can choose a fundamental region $R$ for the action of $\Gamma$ in $\mathbb{U}^3$ in the following way.
$R$ is   bounded by the hemispheres which intersect $\hat{\C}$
along the circles that are orthogonal to $\RR$ and bound the fundamental domain $F$. The fundamental region $R$ is a three-dimensional $CW$-complex
with a single 3-cell given by the interior of $R$. The 2-cells --- the  faces $D_k,
D_k', E_k$, $E_k'$, $k=1, \dots, g$, $G_i$, $G_i'$, $i=1, \ldots, n$, $H_j$, $H_j'$, $j=1, \ldots, r$ are given by the parts of the boundary of $R$
bounded by the intersections of the hemispheres and the arcs $a_k - \bar{a}_k$, $
a_k'-\bar{a}_k'$, $b_k-\bar{b}_k$, $b_k'-\bar{b}_k'$, $c_i-\bar{c}_i$, $c_i'-\bar{c}_i'$, $d_j-\bar{d}_j$, $d_j'-\bar{d}_j'$ respectively. The
1-cells --- the edges, are given by the 1-cells of $F_1-F_2$ and by $e_k^0, e^1_k,
f_k^0, f^1_k$, $s_k$, $k=1, \dots, g$; $g_i^0$, $i=1, \ldots, n$ and $h_j^0, h_j^1$, $j=1, \ldots, r$, defined as follows. The edges $e^0_k$ are
intersections of the faces $E_{k-1}$ and $D_k$ joining the vertices $\bar{a}_k(0)$ to
$a_k(0)$, the edges $e_k^1$ are intersections of the faces $D_k$ and $E_k'$ joining
the vertices $\bar{a}_k(1)$ to $a_k(1)$; $f^0_k=e^0_{k+1}$ are intersections of $E_k$
and $D_{k+1}$ joining $\bar{b}_k(0)$ to $b_k(0)$, $f_k^1$ are intersections of $D_k'$
and $E_k$ joining $\bar{b}_k(1)$ to $b_k(1)$; $s_k$ are intersections of $E'_k$
and $D'_k$ joining $\bar{a}'_k(1)$ to $a'_k(1)$; $g_j^0$ are intersections of $G_{j-1}'$ and $G_j$ joining $\bar{c}_j(0)$ to $c_j(0)$; $h_j^0$ are intersections of $H_{j-1}'$ and $H_j$ joining $\bar{d}_j(0)$ to $d_j(0)$; and $h_j^1$ are intersections of $H_{j}'$ and $H_j$ joining $\bar{d}_j(1)$ to $d_j(1)$. Finally, the 0-cells ---
the vertices, are
given by the vertices of $F$. This property means that the edges of $R$ do not
intersect in $\up^3$. When $\Ga$ is a quasi-Fuchsian group, the fundamental region $R$
is a topological polyhedron homeomorphic to the geodesic polyhedron for the
corresponding Fuchsian group.

As stated in p. 176 of \cite{Kra72} (this proof is given for a Riemann surface with punctures, but its proof works for our case verbatim), one can show that there exists an open set $\mathcal{O}$ in $\overline{\mathbb{U}^3}$ such that $R\subset \mathcal{O}$ and a function $\eta\in C^{\infty}\left(\mathbb{U}^3\cup \Omega\right)$ such that

\begin{enumerate}
\item[(i)] $0<\eta<1$ and $\text{supp}\,\eta\subset \overline{\mathcal{O}}$.

\item[(ii)] For each $Z\in \mathbb{U}^3\cup \Omega$, there is a neighbourhood $U$ of $Z$ and a finite set $\Lambda$ of $\Gamma$ such that $\displaystyle \left.\eta\right|_{\gamma(U)}=0$ for each $\gamma\in\Gamma\setminus \Lambda$.

\item[(iii)] $\displaystyle \sum_{\gamma\in\Gamma}\eta(\gamma Z)=1$ for all $Z\in \mathbb{U}^3\cup\Omega$.

\item[(iv)] The only parabolic fixed points of $\Gamma$ that lie in $\overline{\mathcal{O}}$ are $v_i$, $i=1,\ldots, n$.

\item[(v)] For a fixed constant $\vep_0$, put the region $W_i:=\overline{\mathcal{O}}\cap
\left(\mathcal{H}_{i,\vep_0} \cup  \mathcal{K}_{i,\vep_0}\right)$. Then
$\displaystyle \left.\eta\right|_{\gamma(W_i)}=0$ for $\gamma\in\Gamma\setminus\{ \text{id}, \kappa_i, \kappa_i^{-1}\}$, where $\kappa_i$ is the parabolic transformation with fixed point $v_i$.
\end{enumerate}

\vspace{0.3cm}
For a given $ds^2=e^{\phi(z)}|dz|^2 \in \mathcal{CM}(X\sqcup Y)$, we define
\begin{align*}
\hat{f}(z,t)=\begin{cases} te^{\phi(z)/2}\quad \text{for} \ (z,t)\in \overline{\mathcal{O}} \quad\text{and}\quad t\leq \vep_0/2,\\
 1 \qquad\quad\ \ \text{for}  \ (z,t)\in \overline{\mathcal{O}} \quad\text{and}\quad t \geq \vep_0,
\end{cases}
\end{align*}
and extend it to be a smooth function $\hat{f}$ on $\overline{\mathcal{O}}$. Then let
\begin{align*}
f(Z)=\sum_{\gamma\in\Gamma} \eta(\gamma Z)\hat{f}(\gamma Z),
\end{align*}
which is a $\Gamma$-automorphic function on $\cup_{\gamma\in\Gamma} \gamma R$.
As in \cite{2}, one can show that
\begin{align*}
f(Z)=te^{\phi(z)/2}+O(t^3) \qquad \text{as}\quad t\to 0
\end{align*}uniformly on compact subsets of $ \cup_{\gamma\in\Gamma} \gamma R $. Note that we allow the case
$v_i=\infty$  in the above construction of $f$. Near $\infty$ we have the following fact for $f$.

\begin{lemma}\label{l:no-error}
Over the set $\{ (z,t)\in \mathcal{K}_{i,\vep_0} \,|\, t \leq \vep_0/2\}$ for a sufficiently small $\vep_0>0$,
the level defining function $f$ is given by $f(Z)=te^{\phi(z)/2}$.
\end{lemma}

\begin{proof}
First, let us observe that
\begin{align*}
z(\gamma Z)= \frac{az+b}{d}, \quad t(\gamma Z)=\frac{t}{|d|^2}
\end{align*}
for a parabolic element $\gamma=\begin{pmatrix} a & b\\ 0& d\end{pmatrix}$.
Hence, for the parabolic element $\gamma=\kappa_i^{\pm 1}$,
\begin{align*}
t(\gamma Z) e^{\phi(\gamma Z)/2}=t |d|^{-2} e^{\phi(\gamma z)/2} = t e^{\phi(z)/2},
\end{align*}
where we used $ e^{\phi(\gamma z)/2} = e^{\phi(z)/2}|cz+d|^2$ at the second equality.
Over the region $W_i\cap \{ (z,t)\in \mathcal{K}_{i,\vep_0} \,|\, t \leq \vep_0/2\}$, we have  $\eta|_{\gamma(W_i)}=0$ for $\gamma\in\Gamma\setminus \{\mathrm{id},\kappa_i,\kappa_i^{-1}\}$.
Hence, for $Z\in\{ (z,t)\in \mathcal{K}_{i,\vep_0} \,|\, t \leq \vep_0/2\}$,
\begin{align*}
f(Z)
=&\left(\eta\, t\, e^{\phi/2}\right)(Z) +\left(\eta\, t\, e^{\phi/2}\right)(\kappa_i Z) +\left(\eta\,  t\, e^{\phi/2}\right)(\kappa_i^{-1}Z)\\
=&\left( \eta(Z)  +\eta(\kappa_iZ)  +\eta(\kappa_i^{-1}Z)\right)\, t \,e^{\phi(z)/2}= te^{\phi(z)/2}.
\end{align*}
\end{proof}

Due to the rank one cusps and conical singular lines, some arguments (for instance, the Stokes' Theorem)
do not work properly without truncating some parts near rank one cusps and conical singular lines.
Hence, we truncate a noncompact domain $R$ using the level defining function $f$ near the bottom boundary of $R$ in addition to
removing parts near rank one cusps and conical singularity lines as follows:
\begin{align*}
R_{\vep}=R\cap \{f\geq \vep\}\setminus  \left(\cup_{i=1}^n  \mathcal{P}_{i,\vep} \cup_{j=1}^r  \mathcal{Q}_{j,\vep}\right),
\end{align*}
where $\mathcal{P}_{i,\vep}:=\mathcal{H}_{i,\vep}\cup \mathcal{K}_{i,\vep}$.
By  construction, we have
\begin{align*}
\gamma(R_\vep)=\gamma(R)\cap \gamma( \{f\geq \vep\} ) \setminus  \gamma\left(\cup_{i=1}^n  \mathcal{P}_{i,\vep} \cup_{j=1}^r
\mathcal{Q}_{j,\vep}\right)
\end{align*}
for $\gamma\in\Gamma$. The boundary of $R_\vep$ produced by truncations consists of
\begin{gather*}
 -F_{\vep}=\pa' R_{\vep}\cap\{f=\vep\},\quad T^c_{i,\vep}:=\pa' R_{\vep}\cap \overline{\mathcal{P}}_{i,\vep}, \quad
T^e_{j,\vep}:=\pa'R_\vep \cap \overline{\mathcal{Q}}_{j,\vep}.
\end{gather*}
Let us note  that $T^c_{i,\vep}$ projects to a subset in a pairing cylinder of the $i$-th rank one cusp corresponding to
$v_i$. Hence the boundary of $R_\vep$ is given by
\begin{align*}
\pa' R_\vep=-F_\vep&+\sum_{k=1}^g \left(D_{k,\vep}-D_{k,\vep}'-E_{k,\vep}+E_{k,\vep}'\right)\\ &+\sum_{i=1}^n\left(G_{i,\vep}-G_{i,\vep}'\right)+\sum_{i=1}^n T^c_{i,\vep}\\
\ &+\sum_{j=1}^r\left(H_{j,\vep}-H_{j,\vep}'\right)+\sum_{j=1}^r T^e_{j,\vep}.
\end{align*}
Note that $D_{k,\vep}$, $D_{k,\vep}'$, $E_{k,\vep}$, $E_{k,\vep}'$ are truncated by removing parts
$\{f<\vep\}$, and that $G_{i,\vep}$, $G_{i,\vep}'$ are truncated by removing parts $\{f<\vep\}$ and $\mathcal{P}_{i,\vep}$,
$H_{j,\vep}$, $H_{j,\vep}'$ are truncated by removing parts $\{f<\vep\}$ and
$\mathcal{Q}_{j,\vep}$ respectively. Hence $G_{i,\vep}$, $G_{i,\vep}'$ have a common boundary
with $T^c_{i,\vep}$, and  $H_{j,\vep}$, $H_{j,\vep}'$ have a common boundary
with $T^e_{j,\vep}$ respectively. Put $g^1_{i,\vep}:=G_{i,\vep}\cap T^c_{i,\vep}$ so that $\kappa_i^{-1} g^1_{i,\vep}=
G_{i,\vep}'\cap T^c_{i,\vep}$, and  put $h^1_{j,\vep}:=H_{j,\vep}\cap T^e_{j,\vep}$ so that $\tau_j^{-1} h^1_{j,\vep}=
H_{j,\vep}'\cap T^e_{j,\vep}$.
Let
$$B_\vep=\sum_{k=1}^g\left(E_{k,\vep}\otimes [ \beta_k]-D_{k,\vep}\otimes [\alpha_k]\right)-\sum_{i=1}^n G_{i,\vep}\otimes [\kappa_i]
-\sum_{j=1}^r H_{j,\vep}\otimes [\tau_j] $$
Then we have
\begin{align*}
\pa' R_\vep= -F_\vep+  T^c_{\vep}+ T^e_{\vep}+\pa'' B_\vep,
\end{align*}where\begin{align*}
&T^c_\vep=\sum_{i=1}^n T^c_{i,\vep}, \qquad
  T^e_\vep=\sum_{j=1}^r T^e_{j,\vep}.
\end{align*}
Let $L_\vep$ denote the truncated object for $L$. Then
\begin{align*}
\pa'B_\vep=L_\vep&-\sum_{k=1}^g \left(\left(f_{k,\vep}^1-f_{k,\vep}^0\right)\otimes [\beta_k]-\left(e_{k,\vep}^1-e_{k,\vep}^0\right)\otimes [\alpha_k]\right)\\
&-\sum_{i=1}^n\left(g_{i,\vep}^0\otimes [\kappa_i]\right)+\sum_{i=1}^n\left(g_{i,\vep}^1\otimes [\kappa_i]\right)\\
&-\sum_{j=1}^r\left(h_{j,\vep}^0\otimes [\tau_j]\right)+\sum_{j=1}^r\left(h_{j,\vep}^1\otimes [\tau_j]\right).
\end{align*}
Let
\begin{align*}
E_\vep=&\sum_{k=1}^g \left(e_{k,\vep}^0\otimes[\alpha_k|\beta_k] -
f_{k,\vep}^0\otimes[\beta_k|\alpha_k] +
f_{k,\vep}^0\otimes\left[\gamma_k^{-1}|\alpha_k\beta_k\right]\right)\\ & \;\;\;
-\sum_{k=1}^{g-1}
f_{g,\vep}^0\otimes\left[\gamma_g^{-1}\ldots\gamma_{k+1}^{-1}|\gamma_k^{-1}\right]+\sum_{i=1}^{n-1}g_{1,\vep}^0\otimes[\kappa_1
\cdots\kappa_i|\kappa_{i+1}]\\
&\;\;\;\;\;\;+\sum_{j=1}^{r-1}g_{1,\vep}^0\otimes[\kappa_1
\cdots\kappa_n\tau_1\cdots \tau_j|\tau_{j+1}].
\end{align*}
Then
\begin{align*}
\pa'B_\vep=L_\vep +L^c_\vep+L^e_\vep  - \pa'' E_\vep,
\end{align*}where
\begin{align*}
 L^c_\vep=\sum_{i=1}^n g_{i,\vep}^1\otimes[\kappa_i],  \qquad L^e_\vep=\sum_{j=1}^r h_{j,\vep}^1\otimes[\tau_j].
\end{align*}
We also have
$$\pa' E_\vep=V_\vep,$$
where
\begin{equation*} \label{V}\begin{split}
{V_\vep} = \sum_{k=1}^g& \Big( \,  ({a}_{k,\vep}(0)-\bar{a}_{k,\vep}(0))\otimes[{\alpha}_k|{\beta}_k]\\
 &\ -({b}_{k,\vep}(0)-\bar{b}_{k,\vep}(0))\otimes\left[{\beta}_k |{\alpha}_k\right]\\
&\ +({b}_{k,\vep}(0)-\bar{b}_{k,\vep}(0))\otimes\left[{\gamma}_k^{-1}|{\alpha}_k{\beta}_k\right]\, \Big)\\
 -\sum_{k=1}^{g-1}&\
({b}_{g,\vep}(0)-\bar{b}_{g,\vep}(0))\otimes\left[{\gamma}_g^{-1}\ldots{\gamma}_{k+1}^{-1}|{\gamma}_k^{-1}\right]\\
 +\sum_{i=1}^{n-1}&\ ({c}_{1,\vep}(0)-\bar{c}_{1,\vep}(0))\otimes\left[{\kappa}_1
\cdots{\kappa}_i|{\kappa}_{i+1}\right]\\
+\sum_{j=1}^{r-1}&\ ({c}_{1,\vep}(0)-\bar{c}_{1,\vep}(0))\otimes\left[{\kappa}_1
\cdots{\kappa}_n \tau_1\cdots \tau_j|{\tau}_{j+1}\right].
\end{split}
\end{equation*}

\subsection{Renormalized volume and holography}\label{s1}

As in the section 5 of \cite{2}, let
$$w_3=\frac{1}{t^3}dx\wedge dy\wedge dt$$ be the hyperbolic volume form in $\mathbb{U}^3$, and define

\begin{align*}
w_2=&-\frac{i}{4t^2}dz\wedge d\bar{z},\\
(w_1)_{\gamma^{-1}}=&-\frac{i}{8}\log\left(|ct|^2J_{\gamma}(Z)\right)\left(\frac{\gamma''}{\gamma'}dz-\frac{\overline{\gamma''}}{\overline{\gamma'}}d\bar{z}\right),
\end{align*}so that
\begin{align*}
w_3=dw_2, \quad dw_1=\delta w_2.
\end{align*}
Here $c=c(\gamma)$ and   $(w_1)_{\gamma^{-1}}=0$ if $c(\gamma)=0$.
Since $\mathbb{U}^3$ is simply connected and $\delta w_1$ is closed, there exists $w_0$ such that $\delta w_1=dw_0$. As in \cite{2},  we can choose $w_0$   so that $\delta w_0=0$.

Denote by $V_{\vep}[\phi]$ the hyperbolic volume of $M_{\vep}$ --- the truncated 3-manifold of $\Gamma\backslash\mathbb{U}^3$ obtained by identifying appropriate faces and edges of $R_{\vep}$.  Then
\begin{align}\label{e:volume-comp}
V_{\vep}[\phi]=&\langle w_3, R_{\vep}\rangle \notag \\
=&\langle dw_2, R_{\vep}\rangle \notag \\
=&\langle w_2, -F_{\vep}+T_{\vep}^c+T_{\vep}^e+\pa''B_{\vep}\rangle \notag \\
=&-\langle w_2, F_{\vep}\rangle +\langle w_2, T^c_\vep \rangle +\langle w_2, T^e_\vep \rangle+\langle \delta w_2, B_{\vep}\rangle \notag \\
=&-\langle w_2, F_{\vep}\rangle +\langle w_2, T^c_\vep \rangle +\langle w_2, T^e_\vep \rangle+\langle w_1, L_{\vep}+L_{\vep}^c+L_{\vep}^e-\pa''E_{\vep}\rangle \notag\\
=&-\langle w_2, F_{\vep}\rangle +\langle w_1, L_{\vep}\rangle -\langle w_0, V_{\vep}\rangle \\
&+\langle w_2, T^c_\vep \rangle + \langle w_1, L^c_\vep \rangle+\langle w_2, T^e_\vep \rangle + \langle w_1, L^e_\vep \rangle.\notag
\end{align}
Denote by $A_{\vep}[\phi]$ the area of $\pa'M_{\vep}$ in the induced metric. When $\vep\rightarrow 0^+$, as in the proof of Theorem 5.1 of \cite{2},
one can show that
\begin{align*}
\langle w_2, F_{\vep}\rangle =&-\frac{1}{2}A_{\vep}[\phi] +\frac{i}{8}\langle \omega[\phi], F\rangle -\frac{1}{4}\iint\limits_{F}e^{\phi(z)}d^2z+o(1),\\
\langle w_1, L_{\vep}\rangle =&2\pi\chi(X)\left(\log 2-\log\vep\right)+\frac{i}{8}\langle \check{\theta}[\phi], L\rangle+o(1),\\
\langle w_0, V_{\vep}\rangle=& \frac{i}{8}\langle \check{u}, W\rangle +o(1).
\end{align*} Here $\chi(X)$ is defined in \eqref{e:def-chi}.
Combining \eqref{e:volume-comp} and these equalities,
\begin{align}
V_{\vep}[\phi]-\frac{1}{2}A_{\vep}[\phi]=&\, 2\pi\chi(X)\left(\log 2-\log\vep\right) \notag \\
&\, -\frac{i}{8}\left(\langle \omega[\phi], F\rangle -\langle \check{\theta}[\phi], L\rangle+\langle \check{u}, W\rangle \right)+\frac{1}{4}\iint\limits_{F}e^{\phi(z)}d^2z \notag \\
&\, +\langle w_2, T^c_\vep \rangle + \langle w_1, L^c_\vep \rangle+\langle w_2, T^e_\vep \rangle + \langle w_1, L^e_\vep \rangle\notag +o(1)\\
=&\, 2\pi\chi(X)\left(\log 2-\log\vep \right)-\frac14 \left( S[\phi]-\iint\limits_{F}e^{\phi(z)}d^2z\right)\\
&\, +\langle w_2, T^c_\vep \rangle + \langle w_1, L^c_\vep \rangle+\langle w_2, T^e_\vep \rangle + \langle w_1, L^e_\vep \rangle +o(1). \notag
\end{align}
In the following, we are going to compute the terms $\langle w_2, T^c_\vep \rangle$, $\langle w_1, L^c_\vep \rangle$ from rank one cusps,
and the terms  $\langle w_2, T^e_\vep \rangle$, $\langle w_1, L^e_\vep \rangle$ from the conical singularities.

\subsubsection{Computation of $\langle w_2, T^c_\vep \rangle$}
Here we deal with only the case when $v_i$ is finite since the other case is easier.
For  $\gamma\in \mathrm{PSL}(2,\mathbb{C})$,
\begin{align*}
(\delta w_2)_{\gamma^{-1}} =& \gamma^* w_2- w_2\\
=&\frac{i}{2} J_\gamma(Z) \left( |c|^2 dz\wedge d\bar{z} -\frac{c(\overline{cz+d})}{t} dz\wedge dt + \frac{\bar{c}(cz+d)}{t} d\bar{z}\wedge dt \right).
\end{align*}
Using this with $\gamma=\sigma_i$, we have
\begin{align*}
\langle w_2, T^c_{i,\vep} \rangle  =& \langle \sigma^*_i w_{2}, T^c_{0,\vep} \rangle
= \langle (\delta w_2)_{\sigma_i^{-1}} +w_2, T^c_{0,\vep} \rangle\\
=&\frac{i}{2}\Big{ \langle} \left(\frac{1}{|z|^2+t^2} -\frac{1}{2t^2}\right) dz\wedge d\bar{z}- \frac{\bar{z}dz-zd\bar{z}}{|z|^2+t^2}\wedge \frac{dt}{t}, T^c_{0,\vep} \Big{\rangle},
\end{align*}
where $T^c_{0,\vep}=\hat{T}_\vep \cup \check{T}_\vep$,
\begin{align*}
\hat{T}_{\vep}:=&\{ (z, t)\, | \, z=x+iy, 0\leq x \leq 1,  |y| \leq \frac{1}{\sqrt{\vep}} , t=\frac1{\vep}\}\\
\check{T}_{\vep}:=& \{(z,t)\, | \, z=x+iy, 0\leq x \leq 1, |y|=\pm \frac{1}{\sqrt{\vep}}, \sqrt{\vep} \leq t \leq \frac{1}{\vep} \}.
\end{align*}
Note that we used Lemma \ref{l:no-error} to determine the domain $\check{T}_\vep$. Then,
\begin{align*}
&\Big{\vert} \langle w_2, T^c_{i,\vep} \rangle \Big{\vert} \\
&\leq \Big{\vert}\frac{i}{2} \iint_{\hat{T}_\vep}  \left(\frac{1}{|z|^2+t^2}-\frac{1}{2t^2}\right) \, dz\wedge d\bar{z} \Big{\vert}
+\Big{\vert} \frac{i}{2} \iint_{\check{T}_\vep} \frac{1}{|z|^2 +t^2}\, 2i ydx \wedge \frac{dt}{t} \Big{\vert}\\
& \leq c ( \vep^{3/2} + \vep^{1/2}\log \vep )
\end{align*}
for a constant $c>0$. Hence, $\langle w_2, T^c_{\vep} \rangle$ does not contribute as $\vep\to 0^+$.

\subsubsection{Computation of  $\langle w_2, T^e_\vep \rangle$}
Recall that $\rho$ in \eqref{e:def-rho-lambda}
maps the $t$-axes to the rotation axes of the elliptic element $\tau$. For our purpose, we may assume that $\alpha=0$ in the expression of $\rho$ in \eqref{e:def-rho-lambda}. For such a $\rho$,
\begin{align*}
J_\rho(Z)= \frac{|\mathrm{w}_1-\mathrm{w}_2|}{ |z- 1|^2+t^2},   \quad
-c(\overline{cz+d})= \frac{\bar{z}-1}{|\mathrm{w}_1-\mathrm{w}_2|} .
\end{align*}
Using these, we find that
\begin{align*}
\langle w_2, T^e_{i,\vep} \rangle  =& \langle \rho^*_{i} w_{2}, T^e_{0,\vep} \rangle
\\=& \langle (\delta w_2)_{\rho_{i}^{-1}} +w_2, T^e_{0,\vep} \rangle\\
=&  \frac{i}{2}\Big{\langle}\left(\frac{1}{|z-1|^2+t^2} -\frac{1}{2t^2}\right) dz\wedge d\bar{z}, T^e_{0,\vep} \Big{\rangle}\\
&+\frac{i}{2}\Big{\langle} \left(\frac{ (\bar{z} -1) dz}{ |z-1|^2+t^2} -  \frac{ ({z} -1) d\bar{z}}{ |z-1|^2+t^2} \right) \wedge \frac{dt}{t}\, , \ T^e_{0,\vep} \Big{\rangle}.
\end{align*}
Here $T^e_{0,\vep}$ is a subset in the surface $\{(z,t)\, | \, t=\vep^{-1}|z| \}$ with $0\leq \mathrm{arg}(z) \leq \frac{2\pi}{m}$ and $a\vep \leq  t \leq b\vep^{-1}$ for some $a>0$, $b>0$ by Lemma \ref{l:ell-geo}.
Hence, we have
\begin{align*}
&\Big{\vert} \Big{\langle} \left(\frac{1}{|z-1|^2+t^2} -\frac{1}{2t^2}\right) dz\wedge d\bar{z}\, , \ T^e_{0,\vep} \Big{\rangle} \Big{\vert} \leq c_1 \vep^2\log\vep,\\
&\Big{\vert} \Big{\langle} \left(\frac{ (\bar{z} -1) dz}{ |z-1|^2+t^2} -  \frac{ ({z} -1) d\bar{z}}{ |z-1|^2+t^2} \right) \wedge \frac{dt}{t}\, ,\  T^e_{0,\vep} \Big{\rangle}\Big{\vert}\\
& \leq \ \Big{\vert}\Big{\langle} \frac12\, \frac{ (\bar{z}-1)dz -(z-1)d\bar{z}}{|z-1|t} \wedge \frac{dt}{t},\, \ T^e_{0,\vep} \big{\rangle}\Big{\vert}
\leq c_2 \vep \log\vep
\end{align*}
for some $c_1>0, c_2>0$.
Hence, as $\vep\to 0^+$ the term $\langle w_2, T^e_{\vep} \rangle$ does not contribute.

\subsubsection{Computation of $\langle w_1, L^c\rangle$}
We only deal with the case when $v_i$ is finite since the concerning term is trivial when $v_i=\infty$ by the definition of $(w_1)_{\gamma^{-1}}$.
Recalling $L^c_\vep=\sum_{i=1}^n g_{i,\vep}^1\otimes[\kappa_i]$, we consider
\begin{align*}
&\langle w_1, g_{i,\vep}^1\otimes [\kappa_i] \rangle = \langle (w_1)_{\kappa_i}, g_{i,\vep}^{1} \rangle = \langle \sigma^*_i (w_1)_{\kappa_i}, g_{\vep} \rangle\\
=& \big{\langle} (\delta w_1)_{\sigma_i^{-1},\kappa_i} +(w_1)_{\sigma^{-1}_i\kappa_i} - (w_1)_{\sigma_i^{-1}}\, ,\  g_{\vep} \big{\rangle}.
\end{align*}
Here we may assume that $g_{\vep}=\hat{g}_{\vep}\cup \check{g}_{\vep}$
where
\begin{align*}
\hat{g}_{\vep}=&\left\{\, (z,t)\in\mathbb{U}^3\, | \, z=1+iy, \, | y|
\leq \frac{1}{\sqrt{\vep}}, t=\frac{1}{\vep}\, \right\},\\
\check{g}_{\vep}=&\left\{\, (z,t)\in\mathbb{U}^3\, | \, z=1+iy, \,  y
= \pm \frac{1}{\sqrt{\vep}}, \  \  \sqrt{\vep} \leq t \leq \frac{1}{\vep} \, \right\}
\end{align*}
by Lemma \ref{l:no-error}.
For $(\delta w_1)_{\sigma_i^{-1},\kappa_i}$, we use (5.11) in \cite{2} to obtain
\begin{align*}
&\ (\delta w_1)_{\sigma_i^{-1},\kappa_i}\\
=&\ \frac{i}{4} \left( -\log (|z|^2+t^2) \right)\, \left( d \log | z+1|^2 \right)\\
&-\frac{i}4 \left(-\log(|z+1|^2+t^2)\right)\, \left( d\log |z|^2 \right)\\
&-\frac{i}4\left( \log t^2 -\log (|z|^2+t^2)-\log(|z+1|^2+t^2)\right)
\left(R(z,t)+\bar{R}(z,t) \right),
\end{align*}
where $R(z,t)= (|z|^2+t^2)^{-2} (z+1)^{-1}\left( -\bar{z}|z|^2 dz+ zt^2 d\bar{z}+2t|z|^2 dt\right)$.
Combining this and (5.10) in \cite{2},
\begin{align*}
&(\delta w_1)_{\sigma_i^{-1},\kappa_i} +(w_1)_{\sigma^{-1}_i\kappa_i} - (w_1)_{\sigma_i^{-1}}\\
=&\frac{i}{4} \Phi(z,t) \big(d\log |z +1|^2 - d\log|z|^2 -R(z,t)-\bar{R}(z,t)\big),
\end{align*}
where
$$
\Phi(z,t):=\log t^2 -\log\left({|z+1|^2+t^2}\right)-\log \left({|z|^2+t^2}\right).
$$
Then we
can see that
\begin{align*}
\big{\vert} \Phi(z,t) \big{\vert}  \leq c_1 | \log \vep |
\end{align*}
 over  $g_\vep$, for a constant $c_1>0$.
Using this, it is easy to show
\begin{align*}
\Big{\vert} \int_{g_{\vep}} \Phi(z,t)
(R(z,t)+\bar{R}(z,t)) \Big{\vert} \leq c_2\, \vep |\log\vep|
\end{align*}
for a constant $c_2>0$. The other integrand  $\Phi(z,t) \left(d\log |z +1|^2 - d\log|z|^2\right)$
vanishes over $\check{g}_\vep$. Hence, we need to check the integral of this over $\hat{g}_\vep$. For this, we have
\begin{align*}
 \Phi(z,t) = 2\log \vep - \log\left( \frac{(1+\vep^2(y^2+4))}{(1+\vep^2(y^2+1))}\right)
\end{align*}over $\hat{g}_\vep$.
Hence,
\begin{align*}
&\Big{\vert} \int_{{g}_{\vep}} \Phi(z,t) \left(d\log |z +1|^2 - d\log|z|^2\right) \Big{\vert}\\
\leq &\, \Big{\vert}  \int^{\vep^{-1/2}}_{-\vep^{-1/2}} 2\log\vep   \left(d\log(y^2+4) - d\log(y^2+1) \right) \Big{\vert}\\
& + \Big{\vert} \int^{\vep^{-1/2}}_{-\vep^{-1/2}} \log \left( \frac{(1+\vep^2(y^2+4))}{(1+\vep^2(y^2+1))}\right)
\left(d\log(y^2+4) - d\log(y^2+1) \right) \Big{\vert}\\
&  \leq  c_3\, \vep| \log\vep |
\end{align*}
for a constant $c_3>0$. Let us remark that the integral in the second line vanishes since the one form $d\log(y^2+4) - d\log(y^2+1)$ is odd with
respect to $y$.
Combining the above computations, we have
\begin{align*}
\Big{\vert}  \big{\langle} (\delta w_1)_{\sigma_i^{-1},\kappa_i} +(w_1)_{\sigma^{-1}_i\kappa_i} - (w_1)_{\sigma_i^{-1}}\, , \  g_{\vep} \big{\rangle} \Big{\vert} \leq c\, \vep |\log \vep|
\end{align*}
for a constant $c>0$.
Therefore, the term $\langle w_1, L^c_\vep\rangle$ does not contribute when we take
$\vep\to 0^+$.

\subsubsection{Computation of $\langle w_1, L_\vep^e \rangle$}
Let us recall that
\begin{equation*}
(w_1)_{\gamma^{-1}}=-\frac{i}{8} \log\left( |c t|^2 J_\gamma(Z) \right) \left( \frac{\gamma''}{\gamma'} dz - {\frac{\overline{\gamma''}}{\overline{\gamma'}}}  d\bar{z} \right).
\end{equation*}
 Recall that $h^1_{j,\vep}=H_{j,\vep}\cap T^e_{j,\vep}$ and  we can see that
\begin{equation*}
L^e_\vep=\sum_{j=1}^r h^1_{j,\vep}\otimes[\tau_j]   \  \   \rightarrow  \   \   L^e=\sum_{j=1}^r h_{j}\otimes[\tau_j] \quad \text{as} \quad \vep\to 0^+,
\end{equation*}
where $h_j$ is the geodesic connecting the two fixed points $\mathrm{w}_1$ and $\mathrm{w}_2$ under the action of the elliptic element $\tau_j$. It starts at $\mathrm{w}_2$ and ends at $\mathrm{w}_1$.
\begin{equation*}
\lim_{\vep\rightarrow 0^+}\langle w_1, L^e_{\vep} \rangle = \sum_{j=1}^r \langle (w_1)_{\tau_j}, h_j \rangle.
\end{equation*}
For this geodesic $h_j$ which is the Euclidean semicircle in $\mathbb{U}^3$ perpendicular to $\mathbb{C}$ at $\mathrm{w}_1$ and $\mathrm{w}_2$,
we use the following parametrization of $h_j$:
\begin{align*}
s\in [0,1]\  \ \  \longrightarrow \  \   \    h_j(s)=(z(s), t(s)),
\end{align*}
where $z(s)=s(\mathrm{w}_1-\mathrm{w}_2)+\mathrm{w}_2$ and $t(s)=|\mathrm{w}_1-\mathrm{w}_2|\sqrt{s-s^2}$ satisfying
\begin{equation*}
\left| z(s)- \frac{\mathrm{w}_1+\mathrm{w}_2}{2} \right|^2 + t(s)^2 = \left|\frac{\mathrm{w}_1-\mathrm{w}_2}{2} \right|^2.
\end{equation*}
For the elliptic element $\tau_j$, we have
\begin{align*}
\tau_j^{-1}(z)= \frac{  (\mathrm{w}_1 e^{i\beta} -\mathrm{w}_2 e^{-i\beta} )z -\mathrm{w}_1\mathrm{w}_2 (e^{i\beta} - e^{-i\beta}) }
{ (e^{i\beta}- e^{-i\beta})z-  (\mathrm{w}_2e^{i\beta}- \mathrm{w}_1 e^{-i\beta})},
\end{align*}
where $e^{i\beta}= e^{i \pi/m}$. Hence,
\begin{equation}\label{e:der-tau}
\tau_j^{-1\prime} (z)= \frac { (\mathrm{w}_1-\mathrm{w}_2)^2}{\left( (e^{i\beta}- e^{-i\beta}) z- (\mathrm{w}_2 e^{i\beta} -\mathrm{w}_1 e^{-i\beta}) \right)^2 }.
\end{equation}

\begin{lemma}\label{l:der-logJ} Along the curve $h_j(s)$,
\begin{equation*}
\frac{d}{ds} \log J_{\tau_j^{-1}} (h_j(s))=0.
\end{equation*}
\end{lemma}
\begin{proof}
Notice that
\begin{align*}
&\frac{d}{ds} \log J_{\tau_j^{-1}} (h_j(s))\\
=&-\frac{ c\, z'(s) \, (\bar{c} \bar{z}(s)+ \bar{d}) + \bar{c}\, \bar{z}'(s)\, (cz(s)+d) + 2 t(s) t'(s) |c|^2}
{| c z(s)+ d |^2 + | c t(s)|^2},
\end{align*}
where \begin{gather}
c=c(\tau_j^{-1})= e^{i\beta} - e^{-i\beta}, \hspace{1cm}  d=d(\tau_j^{-1})= \mathrm{w}_1 e^{-i\beta}-\mathrm{w}_2 e^{i\beta},\label{e:express-cd} \\
 z'(s)=\mathrm{w}_1-\mathrm{w}_2,\hspace{2cm} 2t(s)t'(s)=(1-2s)|\mathrm{w}_1-\mathrm{w}_2|^2. \nonumber
\end{gather}
Using
\begin{equation}\label{eq0916_1}
cz(s)+d=(\mathrm{w}_1-\mathrm{w}_2)\left(se^{i\beta}+(1-s)e^{-i\beta}\right),
\end{equation}
one can show that the numerator  vanishes.
\end{proof}

\begin{lemma}\label{l:d-lambda} For $\gamma=\tau^{-1}_j$,  the  equality
\begin{equation*}
 \frac{\gamma''}{\gamma'} dz -{\frac{\overline{\gamma''}}{\overline{\gamma'}}}  d\bar{z} =
-2 d  \log \frac{
 ( e^{i\beta} - e^{-i\beta} )s + e^{-i\beta} } { ( e^{-i\beta} - e^{i\beta} )s + e^{i\beta}}
\end{equation*}holds along the curve $z(s)$.
\end{lemma}
\begin{proof}
This follows from the equality \eqref{e:der-tau} and \eqref{eq0916_1}.
\end{proof}

By Lemma \ref{l:der-logJ} and Lemma \ref{l:d-lambda},
\begin{align*}
&\int^{1-\delta_1}_{\delta_0} (w_1)_{\tau_j}\\
=&\frac{i}{4}\int^{1-\delta_1}_{\delta_0}  \log\left( |c t|^2 J_\gamma(Z) \right)\,
d  \log \frac{
 ( e^{i\beta} - e^{-i\beta} )s + e^{-i\beta} }{ ( e^{-i\beta} - e^{i\beta} )s + e^{i\beta}}  \\
=& \frac{i}{4}\int^{1-\delta_1}_{\delta_0} d \left(  \log\left( |c t|^2 J_\gamma(Z) \right)\,
 \cdot \log \frac{
 ( e^{i\beta} - e^{-i\beta} )s + e^{-i\beta} }  { ( e^{-i\beta} - e^{i\beta} )s + e^{i\beta}}\right)\\
&- \frac{i}{4}\int^{1-\delta_1}_{\delta_0}
 \log \frac{
 ( e^{i\beta} - e^{-i\beta} )s + e^{-i\beta} } { ( e^{-i\beta} - e^{i\beta} )s + e^{i\beta}} \,d   \log\left( |c t|^2  \right)\\
&= (\text{I})+(\text{II})
\end{align*}
for some small $\delta_i >0$ with $i=0,1$ and $\gamma=\tau_j^{-1}$.

\begin{lemma}  As $s\to 0$ or $s\to 1$, $J_{\tau_j^{-1}} (h_j(s)) \to  |\mathrm{w}_1-\mathrm{w}_2|^{-2}$.
\end{lemma}
\begin{proof}
This follows from the equality \eqref{e:express-cd}.
\end{proof}

Now we take $\delta_0$ and $\delta_1$ to be the parameters whose images of the curve $h_j(s)$ meet the hypersurface defined by
$f(z,t)=\vep$. Since the $t$-coordinate of the hypersurface is given by  $t=\vep  e^{-\phi/2} +O(\vep^3)$ and
$t(s)^2=|\mathrm{w}_1-\mathrm{w}_2|^2 (s-s^2)$, we have
\begin{align*}
&\delta_0 = \vep^2 e^{-\phi(\mathrm{w}_2)} |\mathrm{w}_1-\mathrm{w}_2|^{-2} +O(\vep^3) ,\\
 & \delta_1 = \vep^2 e^{-\phi(\mathrm{w}_1)} |\mathrm{w}_1-\mathrm{w}_2|^{-2} +O(\vep^3).
\end{align*}
By these,
\begin{align*}
(\text{I})&= \frac{i}{4} \left( \log ( |c|^2|\mathrm{w}_1-\mathrm{w}_2|^2  J_{\tau_j^{-1}}(Z) )\cdot \log \frac{
  ( e^{i\beta} - e^{-i\beta} )s + e^{-i\beta} } { ( e^{-i\beta} - e^{i\beta} )s + e^{i\beta}} \right) \Big]^{1-\delta1}_{\delta_0}\\
&\ \  +  \frac{i}{4} \left( \log ( s-s^2 )\cdot \log \frac{
  ( e^{i\beta} - e^{-i\beta} )s + e^{-i\beta} } { ( e^{-i\beta} - e^{i\beta} )s + e^{i\beta}} \right) \Big]^{1-\delta1}_{\delta_0}\\
&= -2\beta \left( \log | e^{i\beta}- e^{-i\beta}|  \right) +O(\vep) \\
&\ \ - \frac{\beta}{2} \left( 4\log \vep- 4\log |\mathrm{w}_1-\mathrm{w}_2| - \phi(\mathrm{w}_1) -\phi(\mathrm{w}_2) \right) + O(\vep)\\
&= -\frac{\beta}{2}\left( 4\log \vep - (\phi(\mathrm{w}_1)+\phi(\mathrm{w}_2)) + 4\log | e^{i\beta}- e^{-i\beta}| -4\log |\mathrm{w}_1- \mathrm{w}_2| \right) +O(\vep).
\end{align*}

For the term $(\text{II})$,
\begin{align*}
(\text{II})&= -\frac{i}{4} \int^{1-\delta_1}_{\delta_0} ds \left(\frac{1}{s}- \frac{1}{1-s}\right) \log  \frac{
 ( e^{i\beta} - e^{-i\beta} )s + e^{-i\beta} } { ( e^{-i\beta} - e^{i\beta} )s + e^{i\beta}}\\
&=\frac{i}{4} \int^{1-\delta_1}_{\delta_0} ds \,  2i\beta\, \left(\frac{1}{s}+ \frac{1}{1-s} \right)\\
 &\ \ -\frac{i}{4} \int^{1-\delta_1}_{\delta_0} ds\, \frac{1}{s}\, \log \frac{1+ s(e^{2i\beta} -1)}{ 1+s( e^{-2i\beta}-1)}\\
&\ \ -\frac{i}{4} \int^{1-\delta_1}_{\delta_0} ds\, \frac{1}{1-s}\, \log \frac{1+ (1-s)(e^{2i\beta} -1)}{ 1+(1-s)( e^{-2i\beta}-1)}\\
&= (\text{III})+(\text{IV})+(\text{V}).
\end{align*}
 For $(\text{III})$,
\begin{align*}
(\text{III})= -\frac{\beta}{2}\left( -4\log\vep +(\phi(\mathrm{w}_1)+\phi(\mathrm{w}_2)) +4\log |\mathrm{w}_1-\mathrm{w}_2| \right) +O(\vep).
\end{align*}
Using the definition of dilogarithm function \eqref{0915_1}, we have
\begin{align*}
(\text{IV})+(\text{V})=\frac{i}{2} \left( \mathrm{Li}_2(1- e^{2i\beta})-\mathrm{Li}_2(1-e^{-2i\beta}) \right) +O(\vep).
\end{align*}
Note that the diverging terms in $(\text{I})$ and $(\text{III})$ cancel each other. Hence, we can formulate the concerning integral in terms of the
principal value of the integral. Combining all these computations, we have

\begin{align}\label{e:elliptic-con}
&\langle w_1, L^e \rangle
=\sum_{j=1}^r \mathrm{p.v.} \int_{h_j}  (w_1)_{\tau_j}  \notag \\
&=-\sum_{j=1}^r \, \left( \frac{\pi}{m_j}  \log \left( 4\sin^2\frac{\pi }{m_j} \right)  -\frac{i}{2} \left( \mathrm{Li}_2(1- e^{\frac{2\pi i}{m_j} }) -\mathrm{Li}_2(1-e^{-\frac{2\pi i}{m_j}})  \right)\right) \notag\\
&= -\sum_{j=1}^r \, D\left(1-e^{\frac{2\pi i}{m_j} }\right)
\end{align}
Using the Bloch-Wigner function $D(z)$ \eqref{e:def-BW} and the identity
$$D(z)=-D(1-z),$$we can rewrite the equality \eqref{e:elliptic-con} as follows:

\begin{theorem}
\begin{align*}
\langle w_1, L^e \rangle
=\sum_{j=1}^r \, D\left(e^{\frac{2\pi i}{m_j} }\right).
\end{align*}
\end{theorem}

 \vspace{0.5cm}
For $e^{\phi(z)}|dz|^2 \in \mathcal{CM}(X\sqcup Y)$,
as in \cite{2}, we define the regularized on-shell Einstein-Hilbert action functional as
\begin{align}\label{eq:def-EH}
\mathcal{E}[\phi]=-4\lim_{\vep\rightarrow 0}\left(V_{\vep}[\phi]-\frac{1}{2}A_{\vep}[\phi]+2\pi\chi(X)\log\vep\right).
\end{align}
In other words, $\mathcal{E}[\phi]$ is $-4$ times the renormalized volume of the hyperbolic three manifold $M\simeq \Gamma\backslash\mathbb{U}^3$. The computations above shows that $\mathcal{E}[\phi]$ is the Liouville action
up to some topological data of the quasi-Fuchsian 3-manifold. More precisely,
\begin{theorem}\label{thm:holography}
For a quasi-Fuchsian group $\Gamma$ of type $(g,n; m_1,\ldots, m_r)$ and $e^{\phi(z)}|dz|^2 \in \mathcal{CM}(X\sqcup Y)$
where $X= \Gamma\backslash \Omega_1, Y=\Gamma\backslash \Omega_2$,
\begin{align}\label{eq0916_4}
\mathcal{E}[\phi]=S[\phi]-\iint\limits_{\Gamma\backslash\Omega} e^{\phi(z)}d^2z-8\pi\chi(X)\log 2 -4\sum_{j=1}^r D\left(e^{\frac{2\pi i}{m_j} }\right).
\end{align}
\end{theorem}
Notice that there are contributions from the elliptic fixed points which do not depend on moduli parameters.

\begin{corollary}
The Liouville action functional defined by \eqref{Liouvilleaction} is independent of the choice of fundamental domain.
\end{corollary}

Of particular interest is when $\Gamma$ is a Fuchsian group. Using Theorem \ref{classical_Liouville}, we find that
\begin{theorem}\label{Fuchsian_volume}
When $\Gamma$ is a Fuchsian group and $\phi=\varphi$ is the hyperbolic metric,
\begin{align*}
\mathcal{E}[\varphi]=4\pi \chi(X)(1-2\log 2).
\end{align*}
In other words, the renormalized volume of $M\simeq \Gamma\backslash\mathbb{U}^3$ is equal to $$ \pi\chi(X)(2\log 2-1).$$
\end{theorem}

When $\Gamma$ contains elliptic elements, the appearance of the terms given by the Bloch-Wigner functions in \eqref{eq0916_4} might be seemed intriguing. However, such a term  already appears in the classical Liouville action. In fact, as shown in Theorem \ref{Fuchsian_volume},  this term cancels out when $\Gamma$ is a Fuchsian group and $\phi$ is the hyperbolic metric. In general, the Bloch-Wigner function term is an attribute of the Liouville action when $\Gamma$ contains elliptic elements.

\section{Potential of the TZ metric for an elliptic fixed point} \label{s:elliptic}
Given a quasi-Fuchsian group $\Gamma$ of type $(g, n; m_1, m_2,\ldots, m_r)$, let $\tau_j$ be an elliptic generator with fixed points $\mathrm{w}_{1j}$ and $\mathrm{w}_{2j}$ on $\Omega_1$ and $\Omega_2$ respectively. Corresponding to $\tau_j$, there are  Takhtajan-Zograf metrics on the Teichm\"uller space   $\mathfrak{T}(\Gamma_1)$ for the Riemann surface $X\simeq \Gamma\backslash \Omega_1\simeq \Gamma_1\backslash \mathbb{U}$ and the Teichm\"uller space $\mathfrak{T}(\Gamma_2)$ for the Riemann surface $Y\simeq \Gamma\backslash \Omega_2\simeq \Gamma_2\backslash\mathbb{L}$. They are given by
\begin{align*}
\langle \mu, \nu\rangle^{\text{ell}, 1}_{\text{TZ}, j}=&\iint\limits_{\Gamma\backslash\Omega_1}G(\mathrm{w}_{1j}, z)\mu(z)\overline{\nu(z)}\rho(z)d^2z,\\
\langle \mu, \nu\rangle^{\text{ell}, 2}_{\text{TZ}, j}=&\iint\limits_{\Gamma\backslash\Omega_2}G(\mathrm{w}_{2j}, z)\mu(z)\overline{\nu(z)}\rho(z)d^2z
\end{align*}respectively.  Here $G(z,z')$ denotes the integral kernel of $\left(\Delta_0+\frac12\right)^{-1}$ where
$\Delta_0$ is the hyperbolic Laplacian acting on the space of functions. We refer to \cite{TZ17} for more details about these metrics where a potential function has also been constructed on the Schottky deformation space. In this section, we want to construct a potential function for this metric on the quasi-Fuchsian deformation space.

Consider the function
$$ s_{1j}=\varphi(\mathrm{w}_{1j})=\log \frac{\left|(J_1^{-1})_z(\mathrm{w}_{1j})\right|^2}{\left[\text{Im}\; \left(J_1^{-1}(\mathrm{w}_{1j})\right)\right]^2}$$
 on the Teichm\"uller space   for the Riemann surface $X\simeq \Gamma\backslash \Omega_1$. Here $\mathrm{w}_{1j}$ is chosen so that it varies continuously with respect to moduli parameter.
Since $J_1^{-1}$ is a univalent function  on $\Omega_1$, $(J_1^{-1})_z(\mathrm{w}_{1j})\neq 0$ and this is well-defined.

Choosing a different representative $\tilde{\mathrm{w}}_{1j}=\gamma \mathrm{w}_{1j}$ for some $\gamma\in\Gamma$, we find that
$${s}_{1j}=\varphi(\mathrm{w}_{1j})=\varphi(\tilde{\mathrm{w}}_{1j})+\log |\gamma'(\mathrm{w}_{1j})|^2=\tilde{s}_{1j}+\log |\gamma'(\mathrm{w}_{1j})|^2.$$
Since $\gamma$ varies holomorphically with respect to moduli, we find that
$$L_{\bar{\nu}}L_{\mu}s_{1j}$$ does not depend on the choice of the representative point $\mathrm{w}_{1j}$.

On the Teichm\"uller space  for the Riemann surface $Y\simeq \Gamma\backslash \Omega_2$, one can define the function $s_{2j}=\varphi(\mathrm{w}_{2j})$ in the same way. The same properties as above hold for $s_{2j}$.

On the deformation space $\mathfrak{D}(\Gamma)$,   the function
\begin{align*}
s_j= \varphi(\mathrm{w}_{1j}) +\varphi(\mathrm{w}_{2j}) +2\log|\mathrm{w}_{1j} -\mathrm{w}_{2j}|^2
\end{align*}
does not depend on the choice of representatives. Indeed
for any $\gamma\in\Gamma$, we have
\begin{align*}
&\varphi(\gamma(\mathrm{w}_{1j}))+\varphi(\gamma(\mathrm{w}_{2j}))+2\log\left|\gamma(\mathrm{w}_{1j})-\gamma(\mathrm{w}_{2j})\right|^2\\
=&\varphi(\gamma(\mathrm{w}_{1j}))+\varphi(\gamma(\mathrm{w}_{2j}))+\log|\gamma'(\mathrm{w}_{1j})|^2+\log|\gamma'(\mathrm{w}_{2j})|^2+2\log\left| \mathrm{w}_{1j}- \mathrm{w}_{2j}\right|^2\\
=&\varphi(\mathrm{w}_{1j})+\varphi(\mathrm{w}_{2j})+2\log|\mathrm{w}_{1j}-\mathrm{w}_{2j}|^2.
\end{align*}

\begin{theorem}\label{t:potential-ram}Let $\omega_{\text{TZ}, j}^{\text{ell},i}$ be the symplectic form of the Takhtajan-Zograf metric on the Teichm\"uller space $\mathfrak{T}(\Gamma_i)$ for $i=1,2$ corresponding to the elliptic element $\tau_j$. Then
\begin{align*}
\bar{\pa}\pa s_{1j}=i\omega_{\text{TZ}, j}^{\text{ell},1}, \qquad {\bar{\pa}\pa s_{2j}=i\omega_{\text{TZ}, j}^{\text{ell},2}}.
\end{align*}
Hence, $2s_{j}$ is a well-defined potential over $\mathfrak{D}(\Gamma)\simeq \mathfrak{T}(\Gamma_1)\times \mathfrak{T}(\Gamma_2)$ of the Takhtajan-Zograf metric for the ramification point corresponding to $\tau_j$.
\end{theorem}

\begin{proof}
It suffices to prove that for a harmonic Beltrami differential $\mu$ over $X\simeq \Gamma\backslash \Omega_1$,
\begin{align*}
L_{\bar{\mu}}L_{\mu}s_{1j}=\frac{1}{2}\iint\limits_{\Gamma\backslash\Omega_1}G(\mathrm{w}_{1j}, z)\mu(z)\overline{\mu(z)}\rho(z)d^2z.
\end{align*}
Let $\tilde{\mu}=J_1^{*}\mu$ and let $\hat{\rho}$ be the hyperbolic metric on $\mathbb{U}$. From the commutative diagram
\begin{align*}\begin{CD}
\mathbb{U} @> F^{\vep\tilde{\mu}}>> \mathbb{U} \\
@VVJ_1 V @VVJ_1^{\vep } V  \\
\Omega_1 @> f^{\vep\mu}>> \Omega_1^{\vep\mu}
\end{CD}\end{align*}
we have
\begin{align*}
\varphi^{\vep }\circ f^{\vep\mu}+\log \left| f^{\vep\mu}_z\right|^2=\log\left( F^{\vep\tilde{\mu} *}\hat{\rho}\right)\circ J_1^{-1} + \log\left|\left(J_1^{-1}\right)^{\prime}\right|^2.
\end{align*}From this and the Ahlfors formulae \cite{3}:
\begin{align*}
\left.\frac{\pa}{\pa\vep}\right|_{\vep=0}F^{\vep\tilde{\mu} *}\hat{\rho}=&0,\\
\left.\frac{\pa}{\pa\bar{\vep}}\right|_{\vep=0}F^{\vep\tilde{\mu} *}\hat{\rho}=&0,
\end{align*}
and the Wolpert's formula \cite{Wol}:
\begin{align*}
\left.\frac{\pa^2}{\pa\bar{\vep}\pa\vep}\right|_{\vep=0}F^{\vep \tilde{\mu} *}\hat{\rho}=&\frac{1}{2}\hat{\rho}\left(\tilde{\Delta}_0+\frac{1}{2}\right)^{-1}\left|\tilde{\mu} \right|^2
\end{align*}
where $\tilde{\Delta}_0$ is the hyperbolic Laplacian over $\mathbb{U}$ acting on functions,
we have
\begin{align*}
\left.\frac{\pa^2}{\pa\bar{\vep}\pa\vep}\right|_{\vep=0}s_{1j}^{\vep}=&\left.\frac{\pa^2}{\pa\bar{\vep}\pa\vep}\right|_{\vep=0} \varphi^{\vep}\circ f^{\vep\mu}(\mathrm{w}_{1j})\\
=&\frac{1}{2} \left(\tilde{\Delta}_0+\frac{1}{2}\right)^{-1}\left|\tilde{\mu} \right|^2\circ J_1^{-1}(\mathrm{w}_{1j})\\
=&\frac{1}{2} \left(\Delta_0+\frac{1}{2}\right)^{-1}\left|\mu\right|^2 (\mathrm{w}_{1j})\\
=&\frac{1}{2}\iint\limits_{\Gamma\backslash\Omega_1}G(\mathrm{w}_{1j}, z)\mu(z)\overline{\mu(z)}\rho(z)d^2z.
\end{align*}
This and the similar result for $i=2$ complete the proof of the first claim. The second claim follows from the fact
that $\mathrm{w}_{1j}$ and $\mathrm{w}_{2j}$ varies holomorphically.
\end{proof}

For the geodesic $h_j$ connecting two fixed points $\mathrm{w}_1$ and $\mathrm{w}_2$, its renormalized length is defined by
\begin{equation*}
\mathrm{Length}(h_j)= \mathrm{f.p.}  \lim_{\vep\to 0} \mathrm{length} ( R_\vep \cap h_j),
\end{equation*}
where the length is measured by the induced metric on the geodesic $h_j$ from the hyperbolic metric.
By the isometry carrying $h_j$ to the $t$-axes given by $\rho_j^{-1}$ in \eqref{e:def-rho-lambda}, the renormalized length can be computed by the length of the
corresponding subset of the $t$-axes. By Lemma \ref{l:ell-geo}, we can see that the intersection points
of $h_j$ and the hypersurface $f(z,t)=\vep$ are mapped to the following points at $t$-axis:
\begin{equation*}
t_1=\vep e^{-\varphi(\mathrm{w}_1)/2} |\mathrm{w}_1-\mathrm{w}_2|^{-1} +O(\vep^2), \qquad t_2=\vep^{-1}  e^{\varphi(\mathrm{w}_2)/2} |\mathrm{w}_1-\mathrm{w}_2| +O(\vep^2).
\end{equation*}
Hence, the length of $h_{j, \vep}$ is
\begin{equation*}
\int_{t_1}^{t_2} \frac{dt}{t} = -2\log \vep +\frac12 (\varphi(\mathrm{w}_1)+\varphi(\mathrm{w}_2) ) + 2\log |\mathrm{w}_1-\mathrm{w}_2| +O(\vep).
\end{equation*}
Now we can see that the renormalized length is
\begin{equation*}
\mathrm{Length}(h_j)=\frac12 (\varphi(\mathrm{w}_1)+\varphi(\mathrm{w}_2) ) + 2\log |\mathrm{w}_1-\mathrm{w}_2|,
\end{equation*}
which is just the same as $\frac12 s_j$.
By Theorem \ref{t:potential-ram}, this does not depend on the choice of the fundamental domain, and we have
\begin{theorem}\label{thm:potential-elliptic}
$4\,\mathrm{Length}(h_j) = 2 s_j$ is the potential of the Takhtajan-Zograf metric on $\mathfrak{D}(\Gamma)\simeq \mathfrak{T}(\Gamma_1)\times \mathfrak{T}(\Gamma_2)$ corresponding to the elliptic element $\tau_j$.
\end{theorem}

\appendix
\section{Boundary behavior of the hyperbolic metric}\label{a1}
In this appendix, we collect some facts about the boundary behaviours of $\varphi$.
First we quote some important results about univalent functions on the unit disc (see Theorem 2.4 and Theorem 2.5 in p. 32 of \cite{4}).
\begin{theorem}
If $f:\mathbb{D}\rightarrow\mathbb{C}$ is a univalent function on the unit disc normalized such that $f(0)=0$ and $f'(0)=1$, then
\begin{enumerate}[(a)]
\item $\di  |f'(z)|\leq \frac{1+|z|}{(1-|z|)^3}$,
\item $ \di\left|\frac{zf''(z)}{f'(z)}-\frac{2|z|^2}{1-|z|^2}\right|\leq \frac{4|z|}{1-|z|^2}$.

\end{enumerate}

\end{theorem}

\vspace{0.5cm}
Given an arbitrary univalent function $f:\mathbb{D}\rightarrow\mathbb{C}$, let
$$g(z)=\frac{f(z)-f(0)}{f'(0)}.$$
Then $g(0)=0$ and $g'(0)=1$.
From this, we obtain
\begin{corollary}\label{cor1}
If $f:\mathbb{D}\rightarrow\mathbb{C}$ is a univalent function on the unit disc, then
\begin{enumerate}[(a)]
\item $\di  |f'(z)|\leq \frac{16|f'(0)|}{(1-|z|^2)^3}$,
\item $ \di\left|\frac{zf''(z)}{f'(z)} -\frac{2|z|^2}{1-|z|^2}\right|\leq \frac{4|z|}{1-|z|^2}$.

\end{enumerate}
\end{corollary}

\vspace{0.5cm}
From (a) of Corollary \ref{cor1}, we  obtain
\begin{corollary}\label{cor2}
If $f:\mathbb{D}\rightarrow\mathbb{C}$ is a univalent function on the unit disc, then as $|z|\rightarrow 1^-$,
 $$   |f'(z)|=O\left( \frac{1}{(1-|z|^2)^3}\right).$$\end{corollary}

 \vspace{0.5cm}

Let $\sigma$ be the linear fractional transformation
$$\sigma(z)=\frac{z-i}{z+i}.$$Then $\sigma$ maps $\mathbb{U}$ onto $\mathbb{D}$.
Given a quasi-Fuchsian group $\Gamma$ with domain of discontinuity $\Omega_1\sqcup \Omega_2$, let $\Xi=J_1\circ \sigma^{-1}$. Then
$\Xi$ maps $ \mathbb{D}$ biholomorphically onto $\Omega_1$.   The hyperbolic metric     $e^{\varphi(z)}|dz|^2$ on $\Omega_1$ satisfies
\begin{align*}
e^{\varphi\circ \Xi (z)}|\Xi'(z)|^2=\frac{4}{(1-|z|^2)^2}.
\end{align*}
From this, we find that
\begin{align}\label{eq6_2_1}
\varphi\circ\Xi(z)=&\log 4-2\log (1-|z|^2)-\log |\Xi'(z)|^2,\end{align}
\begin{align}\label{eq3_10_1}
\varphi_z\circ\Xi(z)\Xi'(z)+\frac{\Xi''(z)}{\Xi'(z)}=&\frac{2\bar{z}}{1-|z|^2}.
\end{align}

 \vspace{0.5cm}

\begin{lemma}\label{lemma4}
As $z$ approaches the limit set $\mathcal{C}$,
$$\varphi(z)=O\left(\log \left(1-\left|\Xi^{-1}(z)\right|^2\right)\right)=O\left(\log\left[\left(\text{Im} \, J_1^{-1}(z)\right)^2\right]\right).$$
\end{lemma}
\begin{proof}
From \eqref{eq6_2_1}, we have
\begin{align*}
\varphi(z)=& \log 4-2\log \left(1-|\Xi^{-1}(z)|^2\right)-\log |\Xi'(\Xi^{-1}(z))|^2.
\end{align*}
By Corollary \ref{cor2}, we have
$$\log |\Xi'(\Xi^{-1}(z))|^2=O\left(\log\left(1-|\Xi^{-1}(z)|^2\right)\right)$$The assertion follows from the fact that
\begin{align*}
1-|\sigma(z)|^2=\frac{4\,\text{Im}\,z}{|z+i|^2}.
\end{align*}
\end{proof}

\begin{theorem}\label{thm2}
Let $\Omega_1$ be a component of the domain of discontinuity of the quasi-Fuchsian group $\Gamma$, and let $F_1$ be a fundamental domain for the action of $\Gamma$ on $\Omega_1$. Then the integral
\begin{align*}
\iint\limits_{F_1}\left(|\varphi_z|^2+e^{\varphi}\right)d^2z
\end{align*} is well-defined.
\end{theorem}
\begin{proof}
The group
 $$\hat{\Gamma}_1= \Xi^{-1}\circ \Gamma\circ \Xi  $$ is a subgroup of $\text{PSU}(1,1)$.
 Let $v_1, \ldots, v_n$ be the parabolic fixed points of $\Gamma$ and let $x_i=\Xi^{-1}(v_i)$ be the corresponding parabolic fixed points of $\hat{\Gamma}_1$.  Define $$ F_1^{\vep}=F_1\setminus \bigcup_{i=1}^n \Xi\left(\{|z-x_i|<\varepsilon\}\right).$$
We want to show that the limit
$$\lim_{\varepsilon\rightarrow 0^+}\iint\limits_{F_1^{\vep}}\left(|\varphi_z|^2+e^{\varphi}\right)d^2z$$exists.
Notice that
\begin{align*}
\iint\limits_{F_1^{\vep}}\left(|\varphi_z|^2+e^{\varphi}\right)d^2z=\iint\limits_{\Xi^{-1}(F_1^{\varepsilon})}
\left(|\varphi_z\circ \Xi(z) \Xi'(z)|^2+e^{\varphi\circ \Xi(z)}|\Xi'(z)|^2\right)d^2z.
\end{align*}
By  \eqref{eq3_10_1} and Corollary \ref{cor1},
\begin{align*}
|\varphi_z\circ \Xi(z) \Xi'(z)|=\left|\frac{\Xi''(z)}{\Xi'(z)}-\frac{2\bar{z}}{1-|z|^2}\right|\leq\frac{4}{1-|z|^2}.
\end{align*}
Therefore,
\begin{align*}
\iint\limits_{\Xi^{-1}(F_1^{\varepsilon})}
\left(|\varphi_z\circ \Xi(z) \Xi'(z)|^2+e^{\varphi\circ \Xi(z)}|\Xi'(z)|^2\right)d^2z\leq & 5\iint\limits_{\Xi^{-1}(F_1^{\varepsilon})}
\frac{4}{(1-|z|^2)^2}d^2z.
\end{align*}
As $\varepsilon\rightarrow 0^+$,
$$\iint\limits_{\Xi^{-1}(F_1^{\varepsilon})}
\frac{4}{(1-|z|^2)^2}d^2z$$ gives the hyperbolic area of the Riemann surface $X=\Gamma\backslash\Omega_1$, which is finite. This proves the assertion of the theorem.
\end{proof}
\begin{lemma}\label{lemma1}
Assume that $\mathrm{w}_1$ and $\mathrm{w}_2$ are the fixed points of the elliptic element $\tau\in \text{PSL} (2, \mathbb{C})$ of order $m$, then
\begin{enumerate}[(a)]
\item $\di |\tau'(\mathrm{w}_1)|=|\tau'(\mathrm{w}_2)|=1$,
\item $\di \frac{\tau''(\mathrm{w}_1)}{\tau'(\mathrm{w}_1)}=\frac{2\left(e^{\frac{2\pi i}{m}}-1\right)}{\mathrm{w}_1-\mathrm{w}_2}$,
\item $\di \frac{\tau''(\mathrm{w}_2)}{\tau'(\mathrm{w}_2)}=\frac{2\left(1-e^{-\frac{2\pi i}{m}}\right)}{\mathrm{w}_1-\mathrm{w}_2}$.
\end{enumerate}
\end{lemma}

\begin{proof}
A direct computation from the expression of $\tau$ in \eqref{e:def-rho-lambda} gives the desired result.
\end{proof}

\begin{lemma}\label{lemma3}
Assume that $v\neq \infty$ is a fixed point of the parabolic element $\displaystyle\kappa=\begin{pmatrix} 1+qv & -qv^2\\q & 1-qv\end{pmatrix}\in \text{PSL}(2, \mathbb{C})$, then as $z\rightarrow v$,
\begin{enumerate}[(a)]
\item $\di \kappa'(z)=1+O((z-v))$,\\
\item $\di \frac{\kappa''(z)}{\kappa'(z)}=-2q+O((z-v))$.
\end{enumerate}
\end{lemma}

\begin{proof}
A direct computation from the expression of $\kappa$ in \eqref{e:def-kappa} gives the desired result.
\end{proof}

\begin{theorem}
 The classical Liouville action $S[\varphi]$ is well-defined.
\end{theorem}
\begin{proof}
Recall that the classical Liouville action is defined as
\begin{align*}
S[\varphi] =&\frac{i}{2}\Bigl(\langle \omega[\varphi], F_1-F_2\rangle-\langle\check{\theta}[\varphi], L_1-L_2\rangle+\langle \check{u}, W_1-W_2\rangle\Bigr).
\end{align*}
Theorem \ref{thm2} shows that $\langle \omega[\varphi], F_1-F_2\rangle$ is well-defined. Lemma \ref{lemma4} and Lemma \ref{lemma3} show that $\langle\check{\theta}[\varphi], L_1-L_2\rangle$ is well-defined. Lemma \ref{lemma1} and Lemma \ref{lemma3} show that $\langle \check{u}, W_1-W_2\rangle$ is well-defined.
\end{proof}

In the following, we want to justify the applicability of the Stokes' Theorem in the proof of Theorem \ref{firstvariation}.

Given $\mu\in \Omega^{-1,1}(\Gamma)$, we can write it as
$\mu=\mu_1+\mu_2$, where $\mu_i$ has support on $\overline{\Omega_i}$, $i=1, 2$. Let us just concentrate on $\mu_1$. It can be written as
$$\mu_1(z)=e^{-\varphi(z)}\overline{Q(z)},$$where $Q(z)$ is a cusp form of weight 4 for $\Gamma$. Let $v$ be the fixed point of the parabolic element $\kappa\in \Gamma$. There is a biholomorphism $J: \mathbb{U}\rightarrow\Omega_1$ such that $J(\infty)=v$, and
$$J^{-1}\circ \kappa\circ J=\begin{pmatrix} 1 & 1\\ 0 & 1\end{pmatrix}.$$Notice that
$X\simeq \widetilde{\Gamma}\backslash\mathbb{U}$, where $\widetilde{\Gamma}= J^{-1}\circ\Gamma\circ J$.
$$\tilde{\mu}=\mu_1\circ J \frac{\overline{J'}}{J'}$$is a cusp form of $\widetilde{\Gamma}$. It has an expansion of the form
$$\tilde{\mu}(z)=y^2\sum_{k=1}^{\infty} a_ke^{-2\pi i k\bar{z}}.$$
Therefore,
\begin{align*}
\mu_1(z)=&\left(\text{Im}\,J^{-1}(z)\right)^2\sum_{k=1}^{\infty}a_k e^{-2\pi i k \overline{J^{-1}(z)}}\;\frac{\overline{J^{-1\prime}(z)}}{J^{-1\prime}(z)}.
\end{align*}Hence,
\begin{align}\label{eq7_10_1}\left|\mu_1(z)\right|\sim C\left(\text{Im}\,J^{-1}(z)\right)^2 \exp\left(-2\pi \text{Im}\,J^{-1}(z)\right)\rightarrow 0 \end{align}as $z\rightarrow v$.
\begin{lemma}\label{lemma6}Let $v$ be the fixed point of the parabolic element $\kappa\in \Gamma$. Then
 $\dot{f}_{z\bar{z}}(z)\rightarrow0$ as $z\rightarrow v$.\end{lemma}
\begin{proof}
Recall that $\dot{f}_{\bar{z}}(z)=\mu(z)=\mu_1(z)+\mu_2(z)$. Then, by this and \eqref{eq11_14_5},
\begin{align*}
\dot{f}_{z\bar{z}}(z)=-\varphi_z(z)\mu_1(z)-\varphi_z(z)\mu_2(z).
\end{align*}
Lemma \ref{lemma4} and \eqref{eq7_10_1} show that $\dot{f}_{z\bar{z}}\rightarrow 0$ as $z\rightarrow v$.
\end{proof}

Now we consider $\dot{f}(z)$, $\dot{f}_z(z)$ and $\dot{f}_{zz}(z)$. Recall that
\begin{align*}
\dot{f}(z)=&-\frac{1}{\pi}\iint\limits_{\Omega_1}\frac{\mu_1(\zeta)}{(\zeta-z)}\frac{z(z-1)}{\zeta(\zeta-1)}d^2\zeta
-\frac{1}{\pi}\iint\limits_{\Omega_2}\frac{\mu_2(\zeta)}{(\zeta-z)}\frac{z(z-1)}{\zeta(\zeta-1)}d^2\zeta\\
=&A(z)+B(z).
\end{align*}It suffices to consider $A(z)$.
\begin{align*}
A(z)=&-\frac{1}{\pi}\iint\limits_{|\zeta-z|\leq \epsilon}\frac{\mu_1(\zeta)}{(\zeta-z)}\frac{z(z-1)}{\zeta(\zeta-1)}d^2\zeta-\frac{1}{\pi}\iint\limits_{ |\zeta-z|\geq \varepsilon}\frac{\mu_1(\zeta)}{(\zeta-z)}\frac{z(z-1)}{\zeta(\zeta-1)}d^2\zeta,\\
A_z(z)=&-\frac{1}{\pi}\iint\limits_{|\zeta-z|\leq \epsilon} \mu_1(\zeta)\left(\frac{1}{(\zeta-z)^2}-\frac{1}{\zeta(\zeta-1)}\right)d^2\zeta\\&-\frac{1}{\pi}\iint\limits_{ |\zeta-z|\geq \varepsilon}\mu_1(\zeta)\left(\frac{1}{(\zeta-z)^2}-\frac{1}{\zeta(\zeta-1)}\right)d^2\zeta,\\
A_{zz}(z)=&-\frac{2}{\pi}\iint\limits_{|\zeta-z|\leq \epsilon} \frac{\mu_1(\zeta) }{(\zeta-z)^3} d^2\zeta -\frac{2}{\pi}\iint\limits_{ |\zeta-z|\geq \varepsilon}\frac{\mu_1(\zeta) }{(\zeta-z)^3} d^2\zeta.
\end{align*}These together with \eqref{eq7_10_1}  imply that $A(z)$, $A_z(z)$ and $A_{zz}(z)$ are bounded when $z\rightarrow v$. Hence,

\begin{lemma}\label{lemma5}Let $v$ be the fixed point of the parabolic element $\kappa\in \Gamma$. Then
 $\dot{f}(z)$, $\dot{f}_z(z)$ and $\dot{f}_{zz}(z)$ are bounded when $z\rightarrow v$.\end{lemma}

Using the notation in Theorem \ref{firstvariation}, we have
\begin{theorem}\label{t:justification-Stokes}
$$\left\langle d\xi, F_1-F_2\right\rangle=\left\langle \xi, \pa'(F_1-F_2)\right\rangle.$$
\end{theorem}
\begin{proof}
\begin{align*}
\xi=&-2 \dot{f}_{z\bar{z}}d\bar{z}-\varphi \dot{f}_{z\bar{z}}d\bar{z}-\varphi\dot{f}_{zz}dz.
\end{align*}Lemma \ref{lemma4}, Lemma \ref{lemma6} and Lemma \ref{lemma5} prove the assertion.
\end{proof}

\end{document}